%% file: droid_id.tex
\documentclass[12pt]{article}
\usepackage{setspace}
\usepackage[authoryear,round]{natbib}
\input{preamble}
\input{document_style}
\input{numbering}
\input{bibliography_setup}
\title{\Large\bfseries Partial Identification under Imperfect Measurement:\\
  A Distributionally Robust Approach}
\author{Isaac Meza\\[.3em]
  \normalsize Department of Economics, Harvard University}
\date{}
\hypersetup{
  pdftitle={Partial Identification under Imperfect Measurement: A Distributionally Robust Approach},
  pdfauthor={Isaac Meza}
}
\begin{document}
\spacingset{1}
\maketitle
\input{sections/abstract}
\clearpage
\begin{bibunit}
\previewcitationunit{main}
\previewmainspacing
\input{sections/introduction}
\input{sections/model}
\input{sections/main_results}

\clearpage
\input{sections/discussion_and_conclusion}
\clearpage
{\spacingset{1}
\phantomsection
\addcontentsline{toc}{section}{References}
\putbib
}
\end{bibunit}
\clearpage
\appendix
\begin{bibunit}
\previewcitationunit{appendix}
\previewappendixspacing
\input{appendix/supplementary_contents}
\clearpage
\input{appendix/foundations_and_core_proofs}
\input{appendix/coverage_and_benchmark_refinement}
\input{appendix/computing_parameter_sets}
\input{appendix/missing_outcomes}
\input{appendix/thresholded_networks}
\input{appendix/measurement_error_regression}

\clearpage
{\spacingset{1}
\renewcommand{\refname}{References}
\phantomsection
\label{app:references}
\addcontentsline{toc}{section}{References}
\putbib
}
\end{bibunit}
\end{document}

%% file: preamble.tex
\usepackage[margin=1in]{geometry}
\usepackage[T1]{fontenc}
\usepackage[utf8]{inputenc}
\usepackage{lmodern,microtype}
\usepackage{amsmath,amssymb,amsthm,mathtools}
\usepackage{graphicx,booktabs,tabularx,array,enumitem,float,placeins}
\usepackage{algorithm}
\usepackage[noend]{algpseudocode}
\usepackage{needspace}
\usepackage{xcolor,tikz}
\usetikzlibrary{arrows.meta,positioning}
\usepackage{xurl,hyperref}
\providecommand{\citeyearpar}[1]{\cite{#1}}
\definecolor{inkblue}{HTML}{234C60}
\definecolor{revisionpurple}{HTML}{7B2CBF}

\hypersetup{colorlinks=true,linkcolor=inkblue,citecolor=inkblue,urlcolor=inkblue}
\newtheorem{theorem}{Theorem}[section]
\newtheorem{proposition}[theorem]{Proposition}
\newtheorem{lemma}[theorem]{Lemma}
\newtheorem{corollary}[theorem]{Corollary}
\theoremstyle{definition}
\newtheorem{definition}[theorem]{Definition}
\newtheorem{assumption}{Assumption}[section]
\newtheorem{example}[theorem]{Example}
\theoremstyle{remark}
\newtheorem{remark}[theorem]{Remark}
\newcommand{\E}{\mathbb E}
\newcommand{\X}{\mathcal X}
\newcommand{\Qset}{\mathcal Q}

\DeclareMathOperator{\Lip}{Lip}

\numberwithin{equation}{section}
\setlist{itemsep=3pt,topsep=5pt}

%% file: document_style.tex
\usepackage{etoolbox}
\usepackage[font=small]{caption}
\AtBeginEnvironment{thebibliography}{\normalcolor}

\let\cite\citep
\setcitestyle{aysep={,},citesep={;}}
\newcommand{\spacingset}[1]{\setstretch{#1}}
\newcommand{\previewmainspacing}{\spacingset{1.9}}
\newcommand{\previewappendixspacing}{\spacingset{1.5}}

\newcommand{\previewtightdisplays}{%
  \setlength{\abovedisplayskip}{2pt}%
  \setlength{\belowdisplayskip}{2pt}%
  \setlength{\abovedisplayshortskip}{2pt}%
  \setlength{\belowdisplayshortskip}{2pt}%
}
\apptocmd{\normalsize}{\previewtightdisplays}{}{}
\apptocmd{\small}{\previewtightdisplays}{}{}
\apptocmd{\footnotesize}{\previewtightdisplays}{}{}
\AtBeginDocument{\previewtightdisplays}

%% file: numbering.tex
\usepackage{etoolbox}
\newcommand{\previewmainnumbering}{}
\newcommand{\previewappendixnumbering}{}
\makeatletter
\newcommand{\previewequationalias}[1]{%
  \begingroup
  \protected@edef\@currentlabel{\p@equation\theequation}%
  \label{#1}%
  \endgroup
}
\makeatother
\newcommand{\previewnumberedenvironment}[2]{%
  \csletcs{previewappendix#1}{#1}%
  \csletcs{endpreviewappendix#1}{end#1}%
  \newtheorem{previewmain#1}{#2}%
  \csdef{previewmain#1autorefname}{#2}%
  \appto\previewmainnumbering{%
    \csletcs{#1}{previewmain#1}%
    \csletcs{end#1}{endpreviewmain#1}%
  }%
  \appto\previewappendixnumbering{%
    \csletcs{#1}{previewappendix#1}%
    \csletcs{end#1}{endpreviewappendix#1}%
  }%
}
\theoremstyle{plain}
\previewnumberedenvironment{theorem}{Theorem}
\previewnumberedenvironment{proposition}{Proposition}
\previewnumberedenvironment{lemma}{Lemma}
\previewnumberedenvironment{corollary}{Corollary}
\theoremstyle{definition}
\previewnumberedenvironment{definition}{Definition}
\previewnumberedenvironment{assumption}{Assumption}
\previewnumberedenvironment{example}{Example}
\theoremstyle{remark}
\previewnumberedenvironment{remark}{Remark}

\AtBeginDocument{\previewmainnumbering}
\apptocmd{\appendix}{\previewappendixnumbering}{}{%
  \PackageError{numbering}{Could not set appendix numbering}{}%
}

%% file: bibliography_setup.tex
\usepackage{bibunits}
\defaultbibliographystyle{hapalike_mod}
\defaultbibliography{0_refs}

\makeatletter
\newcommand{\previewcitationunit}[1]{%
  \def\@extra@b@citeb{.#1}%
  \def\@extra@binfo{.#1}%
}
\makeatother

%% file: sections/abstract.tex
\begin{abstract}
Imperfect measurements allow different underlying distributions to generate the same observations.
Restrictions linking unobserved components can make it difficult to characterize all compatible distributions and compute the resulting range of parameter values.
We ask how to construct computable parameter bounds and tighten them without additional data or stronger assumptions.
We develop a general framework for partial identification under imperfect measurement based on distributionally robust optimization.
We use measurement restrictions to construct a neighborhood around a benchmark distribution, ensuring that every compatible distribution lies within it.
Optimizing over this neighborhood gives parameter bounds through a common procedure that can be adapted to different measurement problems.
The choice of benchmark affects how informative these bounds are.
We establish conditions under which changing it can tighten the bounds while preserving every parameter value consistent with the same data and assumptions.
Applications to networks, missing outcomes, and regression illustrate how the framework tightens bounds and can recover the full range of compatible values.
\end{abstract}

\vspace{1.5\baselineskip}

%% file: sections/introduction.tex
\section{Introduction}
\label{sec:introduction}

Economic parameters often depend on information that the data reveal only partially.
Outcomes may be recorded only for selected individuals, or a reported network may conceal the strength of interactions.
Different explanations for this missing information can imply different population means or structural effects, even when the distribution of the observations is known.
The sharp identified set collects all values consistent with the observations and maintained assumptions \cite{Manski1989}.
Computing it can be difficult when structural restrictions link unobserved components.
We ask how to construct computable outer bounds from the measurement model and systematically tighten them by selecting a benchmark distribution, while preserving every parameter value allowed by the same observations and identifying assumptions.

Our primary contribution is a general framework for partial identification under imperfect measurement based on distributionally robust optimization (DRO).
Starting from a benchmark latent distribution, such as an imputation, we use measurement restrictions to bound its distance from every compatible distribution.
The benchmark serves as the center of a neighborhood containing all distributions consistent with the data and assumptions.
Intuitively, the Wasserstein distance measures the smallest average change in latent values needed to transform one distribution into another.
The neighborhood's radius sets the limit on this change.
This neighborhood gives computable parameter bounds.
Choosing its center carefully can tighten both the neighborhood and the resulting bounds while retaining every compatible distribution.
Researchers specify the measurement model, the parameter of interest, and methods for solving the required optimization problems.
The framework then follows the same steps to construct bounds, allowing the procedure to be automated and reused across applications.

Our second contribution establishes when changing the benchmark yields bounds that are at least as tight as the original ones and still contain every parameter value consistent with the data and assumptions.
Theorem~\ref{main:recentering} and Corollary~\ref{main:triangle-certificate} provide both a simple rule for ensuring this guarantee and a more general test that allows additional benchmark choices.
We use these results to formulate optimization programs for benchmark selection.

We study scalar parameters defined by population moments, taking the observed law as known.
The main results use a compact latent space, continuous moments, and a closed, convex class of retained laws.
The framework allows nonsmooth and noninvertible measurement rules.
The approach can save computation when compatibility restrictions are easier to use in calibrating the neighborhood than to impose directly in moment optimization.

Our network application shows how to refine bounds on a moment-defined network-exposure slope when restrictions on which hidden links can coexist make the sharp set difficult to compute.
Known limits on unreported link weights give valid initial bounds, and benchmark selection tightens them without additional data or stronger identifying assumptions.\footnote{%
Network-measurement research studies sampled networks \cite{ChandrasekharLewis2016}, censored nominations \cite{Griffith2022}, unobserved links \cite{LewbelQuTang2023}, and network identification from panel outcomes \cite{dePaulaRasulSouza2025}.}
The missing-outcome and regression examples show that the same procedure can also recover the sharp identified set.

\subsection{Related work}
\label{sec:related-work}

Our use of distributional neighborhoods builds on robust bounds for counterfactuals and treatment effects under uncertainty about latent distributions, marginals, or covariate noise \cite{ChristensenConnault2023,GuRussell2024,FanParkXu2025,GuoEtAl2022}.
The question here is how to choose such a neighborhood for a fixed measurement model.
We calibrate its radius to cover every compatible law, making the benchmark a computational choice that can change without altering the identifying assumptions.

This containment requirement connects our construction to random-set, latent-variable, and transport methods that characterize compatibility with incomplete observations \cite{BeresteanuMolchanovMolinari2011,BeresteanuMolchanovMolinari2012, EkelandGalichonHenry2010,Schennach2014}.
Imposing the full compatible class can remain computationally demanding.
We use measurement restrictions to ensure that a simpler class used for moment optimization still contains every compatible law.
The requirement to retain every compatible law also distinguishes our objective from transport-based moment fitting, which seeks a least-cost adjustment satisfying the moments \cite{SchennachStarck2026}.

Turning these compatibility conditions into bounds is the focus of transport methods and automated polynomial optimization, which provide sharp identified sets or certified approximations \cite{FanParkPassShi2025,Voronin2025,FranguridiLiu2026,DuarteEtAl2024}.
Existing approximation and refinement guarantees do not by themselves justify changing a distributional benchmark.
We certify such changes by requiring the new enclosing class to lie inside the original one.
The resulting guarantee of no wider bounds applies to every target moment covered by the framework.

Selecting that benchmark introduces a geometric problem studied by enclosing-ball methods, which choose centers by minimizing a covering radius \cite{MordukhovichNamVillalobos2013,WangEtAl2025}.
A smaller radius alone does not ensure tighter parameter bounds.
Our nesting conditions use the retained restrictions to determine when a change in center also refines the envelope.
They can therefore certify refinements that a comparison of the unrestricted balls would miss.

Our framework links these ideas through benchmark choice: measurement restrictions determine what the envelope must contain, and nesting conditions determine which changes can refine the moment bounds.
The benchmark therefore governs the computational approximation, while the observed law, identifying assumptions, and sharp identified set remain fixed.

\paragraph{Organization.}
Section~\ref{sec:model} presents the model, Section~\ref{sec:main-results} develops and applies the framework, and Section~\ref{sec:conclusion} concludes.
Appendices A--C provide proofs and computational details; D--F develop the missing-outcome, network, and regression applications.

%% file: sections/model.tex
\section{Model: Learning from imperfect measurements}
\label{sec:model}

\subsection{The parameter and the measurement rule}
We study what can be learned about a population parameter when the variables defining it are only partly observed.
Let $X^\star$ denote the complete, latent data and let $P_{X^\star}$ denote their true probability law.
A scalar parameter $\theta_0\in\Theta\subseteq\mathbb R$ is characterized by
\begin{equation}
 \E_{P_{X^\star}}[g_\theta(X^\star)]=0
 \quad\text{if and only if}\quad \theta=\theta_0.
 \label{eq:true-moment}
\end{equation}
The known function $g_\theta$ specifies a population relationship.
For a mean, for example, $g_\theta(X^\star)=\theta-Y$.
A regression slope is defined by an analogous moment involving both the outcome and the regressor.

We observe $X$, generated by a known measurable measurement rule
\begin{equation}
 X=f(X^\star).
 \label{eq:measurement-map}
\end{equation}
The map $f$ may combine a true value with measurement error, discard a missing outcome, or replace an interaction weight by a binary report.
We write $\X^\star$ and $\X$ for the latent and observed spaces, and $P_X$ for the observed probability law.
This paper takes $P_X$ as known and studies \emph{population identification}.
Estimation from a finite sample is a subsequent question.

With random measurement error, we include signal and noise in $X^\star=(\text{signal},\text{noise})$, so $f$ remains deterministic given the complete state, as illustrated in Appendix~\ref{app:measurement-error}.
Appendix~\ref{app:stochastic} gives the alternative formulation with stochastic measurement and specifies the required error restrictions.
A noninvertible $f$ can leave several latent explanations for the same observations.
Whether those explanations imply different parameter values depends on the target and the maintained assumptions.

\subsection{Compatible laws and the identified set}
Let $\mathcal P(\X^\star)$ denote the probability laws on the latent space.
Collect the maintained support and other restrictions in a class $\mathcal C\subseteq\mathcal P(\X^\star)$.
These may include bounds on network weights, restrictions on measurement error, or conditions ensuring that a regression slope is well-defined.
We assume the true law belongs to $\mathcal C$ and retain any conditions needed for the moment to define the target uniquely.
All moments below are assumed to exist.

\begin{definition}[Compatible latent laws]
\label{def:exact-law-set}
The laws consistent with the observed population are
\begin{equation}
 \Qset_f(P_X)=\{Q\in\mathcal C:f_\#Q=P_X\}.
 \label{eq:compatible-laws}
\end{equation}
Here $f_\#Q$ means the distribution obtained by applying $f$ to data with law $Q$.\footnote{Formally, $f_\#Q$ is the pushforward of $Q$ under $f$: $(f_\#Q)(A)=Q(f^{-1}(A))$ for every measurable $A\subseteq\X$.}
Thus a compatible law reproduces the \emph{entire} observed distribution, including the joint distribution of any variables measured without error.
\end{definition}

We assume $\Qset_f(P_X)$ is nonempty when describing its identified parameters.

\begin{definition}[Exact identified set]
\label{def:exact-id}
The exact population identified set is
\begin{equation}
 \Theta_{\mathrm{ID}}(P_X)
 =\{\theta\in\Theta:\exists Q\in\Qset_f(P_X)
                   \text{ such that }\E_Q[g_\theta(X^\star)]=0\}.
 \label{eq:identified-set}
\end{equation}
\end{definition}

Each value in this set has a compatible latent explanation.
Values outside it cannot satisfy the stated model.
This is the meaning of \emph{sharp identification}.
A singleton gives point identification; a larger set gives partial identification.
The set need not be an interval in general.
Figure~\ref{fig:model-geometry} summarizes the two steps: identify admissible laws, then retain the parameter values they can support.
Additional structural restrictions must also be imposed when claiming sharpness for a stronger model.
The network example illustrates this distinction.

\begin{figure}[htbp]
\centering
\makebox[\linewidth][c]{\scalebox{1.1}{%
\begin{tikzpicture}[
    >=Latex,
    line cap=round,
    line join=round,
    font=\Large,
    every node/.style={inner sep=1.5pt}
]
\definecolor{modelteal}{HTML}{157A78}

\draw[black!45,line width=1pt] (1.48,0.68) ellipse [x radius=1.48, y radius=1.10];
\node[text=black!75] at (0.10,1.92) {\large$\Theta$};

\begin{scope}[shift={(1.42,0.36)}, rotate=28]
    \fill[inkblue!21] (0,0) ellipse [x radius=0.84, y radius=0.48];
    \draw[inkblue,dashed, line width=0.7pt] (0,0) ellipse [x radius=0.84, y radius=0.48];
\end{scope}
\node[text=inkblue] at (0.88,0.97) {\normalsize$\Theta_{\mathrm{ID}}$};

\coordinate (theta) at (1.72,0.45);
\fill[inkblue] (theta) circle (1.7pt);
\node[text=inkblue,below left=1pt and 0.3pt of theta] {\small{$\theta_0$}};

\filldraw[fill=black!6,draw=black!45,line width=0.9pt] (4.55,-0.45) rectangle (9.35,1.95);
\node[text=black!75] at (8.82,2.22) {\normalsize$\mathcal{P}(\mathcal{X}^{\star})$};

\begin{scope}[shift={(6.80,0.65)},scale=1.30]
  \input{figures/compatible_class_outline}
\end{scope}
\node[text=inkblue] at (6.72,0.23) {\normalsize$\Qset_f(P_X)$};

\coordinate (pxs) at (7.05,0.82);
\fill[inkblue] (pxs) circle (1.7pt);
\node[text=inkblue,above=1pt of pxs] {\normalsize$P_{X^\star}$};

\filldraw[fill=black!6,draw=black!45,line width=0.9pt] (10.35,-0.45) rectangle (14.65,1.95);
\node[text=black!75] at (14.42,2.22) {\normalsize$\mathcal{P}(\mathcal{X})$};

\coordinate (px) at (11.45,0.70);
\fill[modelteal] (px) circle (1.7pt);
\node[text=modelteal,above right=2pt and 4pt of px] {\normalsize$P_X$};

\draw[black!65,line width=0.9pt]
    (theta) .. controls (3.10,1.55) and (5.10,0.55) .. (6.88,0.82);
\node[text=black!75] at (3.55,2.22) {\small{$\E_Q[g_\theta(X^\star)]=0$}};

\draw[modelteal,->, line width=0.9pt]
    (7.18,0.86) .. controls (8.75,1.18) and (9.95,1.10) .. (11.28,0.75);
\node[text=modelteal] at (9.80,1.3) {\small{$f_\#$}};

\draw[inkblue,dotted, ->, line width=0.9pt]
    (11.34,0.48) .. controls (9.80,0.02) and (8.55,0.04) .. (6.85,0.54);

\end{tikzpicture}
}}
\caption{The solid $f_\#$ arrow maps the true latent law $P_{X^\star}$ to $P_X$.
The dotted arrow denotes the compatible-law correspondence $P_X\mapsto\Qset_f(P_X)$.
The left curve denotes the moment relation $\E_Q[g_\theta(X^\star)]=0$.
The shaded, dashed region is the identified set $\Theta_{\mathrm{ID}}(P_X)$.}
\label{fig:model-geometry}
\end{figure}
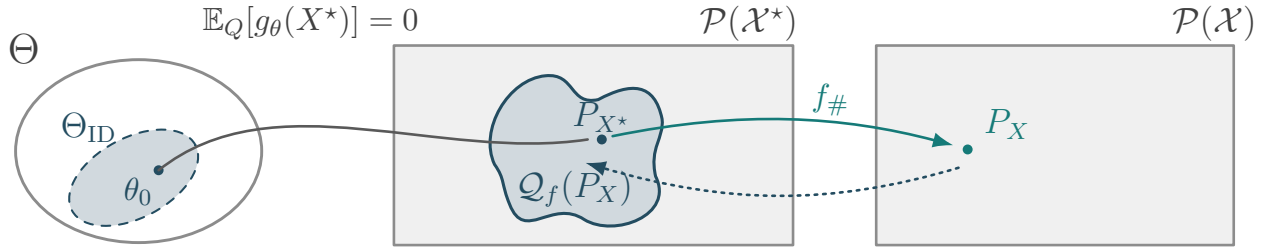

\subsection{Two examples}
\begin{example}[Missing outcomes]
\label{ex:manski-missingness}
Let $D=1$ indicate participation in the workforce and let $Y\in\{0,1\}$ indicate income above a threshold.
Income is observed only for workers.
The target is the population mean $\theta_0=\E[Y]$, with $\Theta=[0,1]$ and
\[
 X^\star=(D,Y),\qquad X=f(D,Y)=(D,DY),\qquad
 g_\theta(D,Y)=\theta-Y.
\]
Both latent states $(0,0)$ and $(0,1)$ produce the observed state $(0,0)$.
Write $p_{ab}=P_X\{(a,b)\}$.
The observed law fixes $p_{10}$ and $p_{11}$, but leaves the success probability among nonworkers unrestricted.
Let $Q_t$ be a possible distribution of $(D,Y)$ that preserves the observed worker outcomes and group proportions, while allocating a fraction $t\in[0,1]$ of the nonworker group to $Y=1$.
Then
\begin{equation}
 \E_{Q_t}Y=p_{11}+p_{00}t,
 \qquad \Theta_{\mathrm{ID}}(P_X)=[p_{11},p_{11}+p_{00}].
 \label{eq:missing-sharp}
\end{equation}
For example, $p_{00}=0.2$ and $p_{11}=0.4$ give the interval $[0.4,0.6]$.
This familiar missing-outcome bound \cite{Manski1990,HorowitzManski2000} is directly computable because a single unknown probability $t$ describes all compatible laws.
It provides a transparent benchmark for evaluating the later robust construction.
The full law and sharpness argument are in Appendix~\ref{app:missingness}.
\end{example}

\Needspace{4\baselineskip}
\begin{example}[Network-exposure slope with a thresholded network]
\label{ex:lim-threshold}
There are $m\ge2$ treated core individuals, indexed $1,\ldots,m$, and two additional individuals: the treated contact $h$ and the untreated contact $s$.
Thus $n=m+2$.
Let $Y\in\mathbb R^n$ collect their outcomes and $T\in\{0,1\}^n$ their treatment indicators, so $T_1=\cdots=T_m=T_h=1$ and $T_s=0$.
The latent network $G^\star\in[0,1]^{n\times n}$ records peer weights: $G^\star_{ij}$ is the weight individual $i$ places on individual $j$'s treatment.

A known set $E$ specifies eligible pairs among the core individuals, oriented from the smaller to the larger index.
Each pair has a known weight $w_{ij}>0$ and capacity $\eta w_{ij}$, where $\eta>0$ is a known scale.
Individual $i$ has an unobserved role $z_i\in\{0,1\}$, distinct from treatment status.
An eligible link is active only between different roles.
Define
\[
 \begin{aligned}
 G^\star_{ij}
   &=\eta w_{ij}\mathbf1\{z_i\ne z_j\}
       &&((i,j)\in E),\\
 G^\star_{is}
   &=1-\sum_{j:(i,j)\in E}G^\star_{ij}
       &&(i=1,\ldots,m),\\
 G^\star_{h1}&=G^\star_{s1}=1,
 \end{aligned}
\]
with all other entries zero.
Each core individual allocates a unit of weight between eligible treated peers and the untreated contact $s$.
An active peer link therefore reduces the weight on $s$.
The last line specifies the rows of the two additional individuals.

Outcomes and treatments are observed without error.
The network is reported only through a known threshold $0<\xi<1$:
\[
 G=\mathbf1\{G^\star\ge\xi\}\quad\text{entrywise}.
\]
Choose $\eta$ so that $\eta\max_{i\le m}\sum_{j:(i,j)\in E}w_{ij}<\min\{\xi,1-\xi\}$.
Then every eligible peer weight is below $\xi$, every core weight on $s$ exceeds $\xi$, and each row sums to one.
For fixed contact labels, every hidden-role assignment therefore gives the same reported network.

A motivating outcome model is
\[
 Y=\beta_0T+\theta_0G^\star T+\varepsilon,
 \qquad \E[\varepsilon\mid G^\star,T]=0.
\]
We take $\beta_0$ as known to isolate network uncertainty.
We maintain the scalar moment below and graph restrictions, but not the conditional mean restriction.\footnote{See Appendix~\ref{app:network}.}

In the general notation, the latent state and measurement map are
\[
 X^\star=(Y,T,G^\star),\qquad
 X=f(X^\star)=(Y,T,\mathbf1\{G^\star\ge\xi\})=(Y,T,G),
\]
and the target $\theta_0\in\Theta=\mathbb R$ is defined by the moment
\begin{equation}
 g_\theta(Y,T,G^\star)
 =T^\top(Y-\beta_0T)-\theta T^\top G^\star T.
 \label{eq:network-moment}
\end{equation}
The maintained class $\mathcal C$ imposes the treatment pattern and graph construction above.
Hidden roles may depend arbitrarily on the observed record.
We take the joint law $P_X$ of $(Y,T,G)$ as known; compatible laws $Q\in\Qset_f(P_X)$ additionally reproduce this entire law.
Consequently,
\[
 \E_Q[g_\theta(X^\star)]
 =\E_{P_X}[T^\top(Y-\beta_0T)]
   -\theta\E_Q[T^\top G^\star T].
\]
The first expectation is known, although outcomes and contact labels may vary across observations.
Exposure $T^\top G^\star T$ is the total network weight between treated units.
The fixed link $G^\star_{h1}=1$ joins two treated individuals, so exposure is at least one.
Its positive expectation ensures that the moment defines a unique slope under each compatible law.

The shared roles restrict which hidden links can coexist.
Maximizing their total exposure can require a difficult global search, although the reported network is unchanged.
The later calculations replace this search with a coverage certificate and constrained moment optimization; explicit endpoint formulas validate the results.
Appendix~\ref{app:network} gives the compatible laws and the numerical specification.
\end{example}

\subsection{Informative bounds from imperfect measurements}
\label{sec:research-question}
Ideally, we would compute the sharp identified set directly.
It consists of all parameter values consistent with the observed law and the maintained measurement model.
This set describes exactly the uncertainty left by imperfect observation.
Our target is the parameter, so recovering this set does not require reconstructing every feature of the latent law.

The difficulty is that compatible laws may be hard to describe or optimize over.
A complicated measurement rule can leave many latent explanations that must reproduce the full observed distribution and satisfy the maintained restrictions.
Searching over these distributions can make exact parameter bounds computationally impractical, even when the target is a single mean or slope.

We propose a relaxation based on distributionally robust optimization (DRO) and transport geometry.
We allow enough variation around a benchmark latent law to include every compatible law, while retaining restrictions that are convenient to enforce.
Optimizing the moment over this structured class gives a \emph{DRO outer identified set} containing the sharp set.
The aim is to make computation more tractable while controlling the extra parameter uncertainty introduced by the relaxation.
Geometrically, this larger class of latent laws forms an envelope containing every compatible law.
A better choice of center can reduce the extra width it creates in parameter space, even if the envelope still contains additional laws.

The question is how to construct and refine computable parameter bounds from the measurement model while preserving every value consistent with the data and maintained assumptions.

%% file: figures/compatible_class_outline.tex
\path[fill=inkblue!21,draw=inkblue,line width=1.05pt]
  (-.80,-.10)
  .. controls (-.95,.22) and (-.67,.69) .. (-.33,.72)
  .. controls (-.05,.75) and (-.11,.44) .. (.16,.50)
  .. controls (.46,.66) and (.68,.66) .. (.75,.32)
  .. controls (.77,.14) and (.63,.06) .. (.69,-.21)
  .. controls (.87,-.57) and (.31,-.78) .. (.03,-.61)
  .. controls (-.27,-.41) and (-.52,-.83) .. (-.66,-.47)
  .. controls (-.74,-.28) and (-.70,-.24) .. (-.80,-.10)
  -- cycle;

%% file: sections/main_results.tex
\section{Main results: Computing the DRO identified set}
\label{sec:main-results}

We answer the question in Section~\ref{sec:research-question} by showing how to compute the parameter bounds implied by the proposed envelope and how to reduce their excess width.
Theorem~\ref{main:feasibility} computes the parameter values supported by laws in the envelope.

Theorem~\ref{main:tube} gives a simpler bound by controlling how far the expected moment can move from its benchmark value.
It also provides an exact alternative calculation when only support and distance restrictions define the envelope.
Theorem~\ref{main:recentering} gives conditions under which changing the benchmark keeps the new envelope inside the original one while retaining every compatible law.
This can tighten the parameter bounds without excluding values consistent with the data and maintained assumptions.

The examples illustrate each result and show when recentering tightens the bounds.
Proofs of the three main theorems and their supporting results are in Appendix~\ref{app:main-proofs}.

\subsection{The ambiguity model and its assumptions}
Let $Q_0$ be a benchmark law on the same space as the candidate latent laws.
For example, imputing values for missing outcomes defines such a law.
When the observed and latent spaces differ, using $P_X$ requires an explicit embedding or projection, as in the regression application in Appendix~\ref{app:measurement-error}.

The Wasserstein distance measures the least average movement between two laws:
\begin{equation}
 W_1(Q,Q_0)=\inf_{\pi\in\Pi(Q,Q_0)}\int d(x^\star,z^\star)
                                      \,\pi(dx^\star,dz^\star).
 \label{eq:wasserstein-distance}
\end{equation}
Here $\Pi(Q,Q_0)$ is the set of joint laws with marginals $Q$ and $Q_0$; each joint law describes how mass moves between latent states.
The distance $d$ specifies which changes are costly.

\Needspace{4\baselineskip}
\begin{assumption}[Compact space and continuous moments]
\label{ass:ambient-regularity}
The latent space $(\X^\star,d)$ is a compact metric space, and $g_\theta:\X^\star\to\mathbb R$ is continuous for each $\theta\in\Theta$.
\end{assumption}

\begin{assumption}[Valid convex relaxation]
\label{ass:retained-class}
The class $\mathcal H\subseteq\mathcal P(\X^\star)$ is weakly closed and convex, contains $\Qset_f(P_X)$, and includes $Q_0$.
It does not depend on $\theta$.
\end{assumption}

The class $\mathcal H$ retains restrictions we can conveniently enforce, such as network support or known observable marginals.
It may be all laws on the chosen space.
Closedness keeps limits admissible; convexity allows mixtures of admissible laws.
Together with compactness, these conditions ensure that the minimum and maximum moments are attained and that every intermediate value is attainable.
They do not require $f$ to be continuous or the exact compatible class to be convex.

For computation, we combine the retained restrictions with a transport budget $\delta\ge0$ around the benchmark:
\begin{align}
 \mathcal A_\delta(Q_0)
 &=\{Q\in\mathcal H:W_1(Q,Q_0)\le\delta\},
 \label{eq:ambiguity-class}\\
 \Theta_{\mathrm{DRO}}(\delta;Q_0)
 &=\{\theta\in\Theta:\exists Q\in\mathcal A_\delta(Q_0),
                               \ \E_Q[g_\theta(X^\star)]=0\}.
 \label{eq:dro-set}
\end{align}
We call $\mathcal A_\delta(Q_0)$ the \emph{envelope}: it is the class of latent laws over which we optimize.
The set $\Theta_{\mathrm{DRO}}(\delta;Q_0)$ collects the parameter values these laws can support.
To calibrate the budget and assess individual candidates, define the covering radius and the feasibility cost, respectively:\footnote{We use the convention $\inf\varnothing=+\infty$.}
\begin{align}
 r^\star(Q_0)&=\sup_{Q\in\Qset_f(P_X)}W_1(Q,Q_0),
 \label{eq:covering-radius}\\
 R(\theta;Q_0)&=\inf_{\substack{Q\in\mathcal H\\\E_Q[g_\theta(X^\star)]=0}}
                   W_1(Q,Q_0).
 \label{eq:feasibility-profile}
\end{align}
A budget $\delta\ge r^\star(Q_0)$ includes every compatible law, so the resulting parameter set contains the sharp identified set.
A computable upper bound on $r^\star$ suffices for this guarantee.
The cost $R(\theta;Q_0)$ measures the least movement from $Q_0$ to a law in $\mathcal H$ satisfying the candidate moment.
Theorem~\ref{main:feasibility} shows that comparing this cost with $\delta$ gives an alternative way to compute the parameter set.

The envelope intersects the benchmark ball with the retained class $\mathcal H$.
Covering all compatible laws preserves their zero-moment witnesses, while the additional laws can introduce additional parameter values.
Figure~\ref{fig:algorithm-recentering} shows this geometry together with benchmark refinement.

\subsection{Two ways to compute feasible parameters}
For a chosen benchmark and radius, the next two results turn the envelope into parameter bounds.
Theorem~\ref{main:feasibility} uses the smallest and largest expected moments allowed by the envelope to determine exactly which parameters belong to the DRO set.
Under an additional smoothness condition, Theorem~\ref{main:tube} gives a simpler outer bound using the benchmark moment and radius.
When support and distance are the only law restrictions, it also gives an exact optimization formula for the moment bounds.
These calculations supply the inputs to the parameter-set computation in Section~\ref{sec:parameter-computation}.

\Needspace{4\baselineskip}
\begin{theorem}[General computation of the DRO set]
\label{main:feasibility}
Under Assumptions~\ref{ass:ambient-regularity} and \ref{ass:retained-class}, the attainable values of $\E_Q[g_\theta(X^\star)]$ over $\mathcal A_\delta(Q_0)$ form a closed interval.
For every $\theta\in\Theta$, the following are equivalent:
\begin{equation}
 \begin{aligned}
 \textnormal{(i)}\quad&\theta\in\Theta_{\mathrm{DRO}}(\delta;Q_0);\\[2pt]
 \textnormal{(ii)}\quad&
 \min_{Q\in\mathcal A_\delta(Q_0)}\E_Q[g_\theta(X^\star)]
 \le0\le
 \max_{Q\in\mathcal A_\delta(Q_0)}\E_Q[g_\theta(X^\star)];\\[2pt]
 \textnormal{(iii)}\quad&R(\theta;Q_0)\le\delta.
 \end{aligned}
 \label{main:feasibility-test}
\end{equation}
If $\delta\ge r^\star(Q_0)$, then $\Theta_{\mathrm{ID}}(P_X)\subseteq \Theta_{\mathrm{DRO}}(\delta;Q_0)$.
If also $\mathcal H=\Qset_f(P_X)$, the two parameter sets are equal.
\end{theorem}

The computation requires two moment optimizations for each candidate $\theta$, or one minimum-cost calculation.
If zero lies between the two bounds, including equality, mixing their laws yields a law with zero moment.
If both bounds are strictly positive or strictly negative, that parameter is ruled out.
This is exact for the chosen DRO model.
With a valid covering radius, any excess over the sharp set comes from additional laws admitted by the envelope.

\Needspace{4\baselineskip}
\begin{assumption}[Moment smoothness]
\label{ass:moment-sensitivity}
For each fixed $\theta\in\Theta$, the moment function is Lipschitz continuous in the latent state, with a known finite constant $L(\theta)\ge0$:
\begin{equation}
 |g_\theta(x^\star)-g_\theta(z^\star)|
 \le L(\theta)d(x^\star,z^\star)
 \quad\text{for all }x^\star,z^\star\in\X^\star.
 \label{eq:moment-lipschitz}
\end{equation}
\end{assumption}

The smoothness condition bounds how far expected moments can move from the benchmark, giving a simpler outer bound without solving the moment optimizations in Theorem~\ref{main:feasibility}.
When only support and distance restrictions are retained, the next result also gives a dual formula for the exact moment bounds.

\Needspace{4\baselineskip}
\begin{theorem}[A simpler bound under smoothness]
\label{main:tube}
Under Assumptions~\ref{ass:ambient-regularity}--\ref{ass:moment-sensitivity}, every $Q\in\mathcal A_\delta(Q_0)$ satisfies $|\E_Q[g_\theta(X^\star)]-\E_{Q_0}[g_\theta(X^\star)]| \le\delta L(\theta)$.
Consequently,
\begin{equation}
 \Theta_{\mathrm{DRO}}(\delta;Q_0)
 \subseteq\mathcal T_\delta(Q_0)
 :=\{\theta\in\Theta:|\E_{Q_0}[g_\theta(X^\star)]|\le\delta L(\theta)\}.
 \label{main:tube-set}
\end{equation}
The name \emph{tube} refers to a band extending $\delta L(\theta)$ above and below the benchmark moment as $\theta$ varies.
The set $\mathcal T_\delta(Q_0)$ consists of the parameter values where this band reaches zero.
For $\mathcal H=\mathcal P(\X^\star)$, the exact upper endpoint also has the representation
\begin{equation}
 \begin{aligned}
 &\max_{Q\in\mathcal A_\delta(Q_0)}\E_Q[g_\theta(X^\star)]\\
 &\quad=\inf_{\lambda\ge0}\left\{\lambda\delta+
  \E_{Z^\star\sim Q_0}\!\left[
    \max_{x^\star\in\X^\star}
     \{g_\theta(x^\star)-\lambda d(x^\star,Z^\star)\}
  \right]\right\}.
 \end{aligned}
 \label{main:dual}
\end{equation}
The lower endpoint is the negative of this expression with $g_\theta$ replaced by $-g_\theta$.
\end{theorem}

The tube $\mathcal T_\delta(Q_0)$ needs only $\E_{Q_0}[g_\theta(X^\star)]$, $\delta$, and $L(\theta)$, so it can be evaluated without optimizing over laws.
It may be wider than the DRO set because some moment values within the band cannot be attained by laws in $\mathcal A_\delta(Q_0)$.

The dual representation in~\eqref{main:dual} replaces infinite-dimensional optimization over probability laws with a scalar convex minimization whose objective involves maximizations over latent states.
Appendix~\ref{app:dual-computation} gives a finite approximation with explicit error bounds.
This representation is an application of established Wasserstein duality \cite{BlanchetMurthy2019,GaoKleywegt2023}.

This dual procedure applies to $\mathcal H=\mathcal P(\X^\star)$.
Additional retained restrictions must remain in the optimization, as in the finite programs of Appendix~\ref{app:finite-program}.
These calculations are at a fixed $\theta$ and determine whether that parameter is feasible.
The parameter set itself need not be an interval.

\subsection{Calibrating the radius and computing the parameter set}
\label{sec:parameter-computation}
For a given benchmark, we first calibrate its radius and then compute the parameter set.
Algorithm~\ref{alg:main-dro-construction} in Section~\ref{sec:recentering} combines these calculations with benchmark choice through recentering.

\paragraph{Calibrate the radius.}
For a chosen $Q_0\in\mathcal H$, any $\delta\ge r^\star(Q_0)$ in \eqref{eq:covering-radius} preserves every compatible parameter value.
A practical way to obtain such a budget is to construct $Q_0$ by imputing unobserved values with a measurable rule $b_0:\X\to\X^\star$, so that $Q_0=(b_0)_\#P_X$.
Couple each compatible latent state with the imputed state from its own observation.
This gives
\begin{equation}
 r^\star(Q_0)\le\overline r(b_0)
 :=\sup_{Q\in\Qset_f(P_X)}
       \E_Q d\bigl(X^\star,b_0(f(X^\star))\bigr).
 \label{eq:calibration-certificate}
\end{equation}
The bound holds because Wasserstein distance is no greater than the cost of this particular coupling.
Computing the supremum can still be difficult.
The operational input is any verified upper bound, which need not equal either $\overline r(b_0)$ or $r^\star(Q_0)$.

For example, if a known, integrable function $h$ satisfies
\[
 d\bigl(x^\star,b_0(f(x^\star))\bigr)\le h(f(x^\star))
 \quad\text{for every admissible latent state},
\]
then $\delta=\E_{P_X}h(X)$ is valid without optimizing over compatible laws.
Maintained moment bounds can also supply a certificate.
Below, $(1-D)Y\le1-D$ gives $\delta=p_{00}$ for missing outcomes.
Appendix~\ref{app:recenter-coupling} proves the general rule, including projected states and moment bounds.

The relaxation saves computation only if we can readily obtain a valid covering radius and compute moment bounds over the resulting class.
The computational cost of a linear program also depends on its size and the work needed to construct it.
The network example below illustrates this distinction: summing known link capacities gives the radius, while obtaining the sharp parameter endpoint requires a global search over hidden roles.
Appendix~\ref{app:radius-network} explains why calibration is simpler in this construction.

\paragraph{Compute the parameter set.}
For a calibrated pair $(Q_0,\delta)$, Theorem~\ref{main:feasibility} characterizes the parameter set through the smallest and largest expected moments over $\mathcal A_\delta(Q_0)$.
The examples compute these moment bounds by optimization; analytical formulas provide independent checks of the calculations.
Appendix~\ref{app:computational-analysis} develops numerical methods when explicit calculations are unavailable.
It supplies affine endpoint search and an interval set-inversion guarantee for a fixed calibrated benchmark and radius.

\subsection{Applying the results to the examples}
\label{sec:dro-examples}
For each moment in Section~\ref{sec:model}, we specify the retained laws $\mathcal H$, choose the benchmark $Q_0$, and compute a radius $\delta$ covering the compatible laws.
These choices define $\mathcal A_\delta(Q_0)$.
Theorem~\ref{main:feasibility} then gives the DRO set from its smallest and largest expected moments.
Theorem~\ref{main:tube} gives the comparison bound using the benchmark moment and $L(\theta)$.
Appendices~\ref{app:missingness} and~\ref{app:network-calculation} supply complete derivations and numerical certificates.

\Needspace{7\baselineskip}
\paragraph{Missing outcomes.}
This example shows how the DRO interval can exceed the sharp bounds by allowing changes to outcomes that are already observed.
Section~\ref{sec:recentering-examples} shows how recentering recovers the sharp set.
Keep $g_\theta(D,Y)=\theta-Y$ from Example~\ref{ex:manski-missingness}.
Assign zero to every missing outcome to obtain $Q_0=Q_{t=0}$.
Take $\mathcal H=\mathcal P(\{0,1\}^2)$ and Hamming distance, which counts changed binary coordinates.
The imputation cost $(1-D)Y\le1-D$ directly certifies the radius $p_{00}$.
Here the certificate is exact: a compatible imputation law $Q_t$ changes a fraction $t$ of the missing outcomes from zero to one.
Moving this mass costs $p_{00}t$, so the radius is
\[
 \delta=r^\star(Q_0)
 =\sup_{t\in[0,1]}W_1(Q_t,Q_0)
 =\sup_{t\in[0,1]}p_{00}t=p_{00}.
\]
The moment optimizations are
\[
 \begin{aligned}
 \min_{Q\in\mathcal A_\delta(Q_0)}\E_Q[g_\theta(D,Y)]
   &=\theta-\min\{1,p_{11}+\delta\},\\
 \max_{Q\in\mathcal A_\delta(Q_0)}\E_Q[g_\theta(D,Y)]
   &=\theta-\max\{0,p_{11}-\delta\}.
 \end{aligned}
\]
Theorem~\ref{main:feasibility} retains the values for which zero lies between these bounds, giving
\begin{equation}
 \Theta_{\mathrm{DRO}}(\delta;Q_0)
 =[\max\{0,p_{11}-\delta\},\ \min\{1,p_{11}+\delta\}].
 \label{ex:missing-dro}
\end{equation}
For Theorem~\ref{main:tube}, the benchmark moment is $\E_{Q_0}[g_\theta(D,Y)]=\theta-p_{11}$ and $L(\theta)=1$.
Its inequality $|\theta-p_{11}|\le\delta$, restricted to $\Theta=[0,1]$, gives the same interval.
Substituting $\delta=p_{00}=0.2$ and $p_{11}=0.4$ gives $[0.2,0.6]$, compared with the sharp set $[0.4,0.6]$ from Section~\ref{sec:model}.
The interval extends below the sharp lower bound because the transport ball allows changes to workers' observed outcomes.
All calculations are verified in Appendix~\ref{app:missing-dro-calculation}.

\paragraph{Thresholded networks.}
This example gives computable outer bounds when obtaining the sharp set requires a difficult search over hidden network roles.
Continue Example~\ref{ex:lim-threshold} with bounded outcomes $Y$.
For the relaxation, include the observed record $X$ alongside $G^\star$, so candidate laws describe $(X,G^\star)$.
Retain the known law $P_X$ and the nonnegative, zero-diagonal graph support $\mathcal S=\{A\ge0:A_{ii}=0,\ A\mathbf1=\mathbf1\}$.
The class $\mathcal H$ preserves $P_X$ and allows any graph in $\mathcal S$.
It relaxes the common-role restriction and the requirement that the candidate graph reproduce the report.
The observed network $G$ has one unit-weight outgoing link per row, so imputing zero for every hidden link gives
\[
 G^0(G)=G,\qquad Q_0=\mathcal L_{P_X}(X,G).
\]
Here $\mathcal L_{P_X}(X,G)$ is the joint law of $X\sim P_X$ and its zero-imputed latent graph $G$.
Let $d_i=\sum_{j:(i,j)\in E}w_{ij}$ and $d_{\max}=\max_i d_i>0$.
We scale graph distances by known row capacities, so rows with less possible hidden weight receive a smaller allowance.
For $x=(y,t,g)$ and $x'=(y',t',g')$, set
\[
 \begin{gathered}
 \omega_i=\frac1{10}+\frac{9d_i}{10d_{\max}}
       \quad(i=1,\ldots,m),\qquad \omega_h=\omega_s=1,\\
 d_{\mathrm{cap}}(A,B)=\max_{i,j}\frac{|A_{ij}-B_{ij}|}{\omega_i},\\
 d((x,A),(x',B))=\|y-y'\|_\infty
 +10\mathbf1\{(t,g)\ne(t',g')\}+d_{\mathrm{cap}}(A,B).
 \end{gathered}
\]
The positive floor keeps the distance well-defined for rows with no eligible links.
The calibrated radius is
\[
 \delta=r^\star(Q_0)
 =\E_{P_X}\!\left[\sup_{A\in\Gamma(G)}d_{\mathrm{cap}}(A,G)\right]
 =\eta d_{\max}<\min\{\xi,1-\xi\}.
\]
The graph diameter under $d_{\mathrm{cap}}$ is at most $10$, making the same-record coupling optimal for $Q_0$; see Appendix~\ref{app:network-calculation}.
Here $\Gamma(G)$ is the set of latent networks that could have produced $G$ subject to the role and weight restrictions in Example~\ref{ex:lim-threshold}.
The radius requires only sums of known capacities.
It does not require finding the role assignment with the greatest total exposure.

For this benchmark, the exposure optimizations have explicit values:
\[
 \begin{aligned}
 \underline b_\delta
  &:=\min_{Q\in\mathcal A_\delta(Q_0)}\E_Q[T^\top G^\star T]
    =1-\delta,\\
 \overline b_\delta
  &:=\max_{Q\in\mathcal A_\delta(Q_0)}\E_Q[T^\top G^\star T]
    =1+\delta\sum_{i=1}^m\omega_i.
 \end{aligned}
\]
The lower bound allows weight to move from the baseline treated link to the untreated contact.
For the upper bound, core row $i$ can move weight $\delta\omega_i$ toward treated peers.
Since $\mathcal H$ retains the observed law, the moment extrema in Theorem~\ref{main:feasibility}, for $\theta\ge0$, are
\[
 \begin{aligned}
 \min_{Q\in\mathcal A_\delta(Q_0)}\E_Q[g_\theta]
  &=\E_{P_X}[T^\top(Y-\beta_0T)]
       -\theta\left(1+\delta\sum_{i=1}^m\omega_i\right),\\
 \max_{Q\in\mathcal A_\delta(Q_0)}\E_Q[g_\theta]
  &=\E_{P_X}[T^\top(Y-\beta_0T)]-\theta(1-\delta).
 \end{aligned}
\]
For a positive observed contribution, the theorem therefore gives
\begin{equation}
 \Theta_{\mathrm{DRO}}(\delta;Q_0)
 =\Theta\cap\left[
 \frac{\E_{P_X}[T^\top(Y-\beta_0T)]}{1+\delta\sum_{i=1}^m\omega_i},\,
 \frac{\E_{P_X}[T^\top(Y-\beta_0T)]}{1-\delta}\right].
 \label{ex:network-dro}
\end{equation}
Since $\mathcal H$ fixes the observed contribution, we apply Theorem~\ref{main:tube} to the equivalent moment $\E_{P_X}[T^\top(Y-\beta_0T)]-\theta T^\top G^\star T$.
Its expectation equals that of $g_\theta$ throughout $\mathcal H$.
Exposure is $(1+\sum_{i=1}^m\omega_i)$-Lipschitz, giving the \emph{exposure tube}\footnote{This is the residual tube in Appendix~\ref{app:recentering-network}.
It need not equal \eqref{main:tube-set} computed from the original moment $g_\theta$.}
\[
 \left|\E_{P_X}[T^\top(Y-\beta_0T)]-\theta\right|
 \le\delta\left(1+\sum_{i=1}^m\omega_i\right)|\theta|.
\]
The factor adds the scaled allowances of the core rows and the treated contact.

The computational gain is a reusable procedure: calibrate from capacities, optimize moments over the retained envelope, and compute parameter bounds.
Its implementation does not require the analytical endpoint formulas displayed here for validation.
The price is possible excess width, which the capacity comparators in Section~\ref{sec:recentering-examples} help diagnose.
Appendix~\ref{app:network-calculation} proves the formulas and connects sharp computation to maximum cut.
Its 152-node numerical illustration is shown in Figure~\ref{fig:dro-example-sets}.

The regression application in Appendix~\ref{app:regression-calculation} shows how a bound on measurement-error size supplies a valid radius without assuming independent errors.
Retaining the observed joint law and squared-error budget makes moment optimization tractable and recovers the sharp set, with no grid for the true regressor.

\Needspace{3\baselineskip}
The examples separate two sources of excess width: replacing exact compatibility by a transport class, and replacing exact transport optimization by the tube.
This distinction guides the choice of benchmark and radius developed next.

\FloatBarrier
\input{sections/recentering}

%% file: sections/recentering.tex
\subsection{Recentering to tighten the outer set}
\label{sec:recentering}
A benchmark obtained by imputation can lie near the edge of the compatible class.
Recentering seeks a benchmark with a smaller covering radius, while keeping the observed law and maintained model fixed.
Coverage and improvement require separate checks: the new envelope must still contain every compatible law, and must lie within the old envelope.
We first give the exact nesting criterion, then a less costly sufficient certificate.

For a candidate $B\in\mathcal H$ and the radius $\delta_B$ actually used, define
\begin{equation}
 \Phi(B,\delta_B;Q_0)
 :=\max_{\substack{Q\in\mathcal H\\W_1(Q,B)\le\delta_B}}W_1(Q,Q_0).
 \label{rec:phi}
\end{equation}
This is the largest distance from the old benchmark among laws allowed by the proposed envelope.
The retained restrictions enter the test and can permit nesting even when the full balls would not be nested.

Changing the benchmark can admit different latent laws even when the covering radius decreases.
The next theorem identifies when recentering preserves coverage and yields a DRO set contained in the initial one.

\begin{theorem}[Coverage and exact envelope nesting]
\label{main:recentering}
Under Assumptions~\ref{ass:ambient-regularity} and \ref{ass:retained-class}, $r^\star$ is convex under mixtures, 1-Lipschitz in its benchmark, and attains a minimum on $\mathcal H$.
For $B,Q_0\in\mathcal H$ and nonnegative radii,
\begin{equation}
 \Phi(B,\delta_B;Q_0)\le\delta_0
 \quad\Longleftrightarrow\quad
 \mathcal A_{\delta_B}(B)\subseteq\mathcal A_{\delta_0}(Q_0).
 \label{rec:exact-nesting}
\end{equation}
If these equivalent nesting conditions hold and $\delta_B\ge r^\star(B)$, then
\begin{equation}
 \Theta_{\mathrm{ID}}(P_X)
 \subseteq\Theta_{\mathrm{DRO}}(\delta_B;B)
 \subseteq\Theta_{\mathrm{DRO}}(\delta_0;Q_0).
 \label{rec:nested-dro}
\end{equation}
\end{theorem}

\Needspace{5\baselineskip}
Evaluating the exact nesting criterion may require a difficult optimization.
The triangle inequality provides a simpler sufficient certificate, which also guarantees tube nesting under Assumption~\ref{ass:moment-sensitivity}.
The triangle test can reject every benchmark change when the initial radius is exact; Algorithm~\ref{alg:main-dro-construction} therefore also checks sharper nesting certificates (Appendix~\ref{app:three-point}).

\begin{corollary}[Triangle certificate for nested bounds]
\label{main:triangle-certificate}
Under Assumptions~\ref{ass:ambient-regularity} and \ref{ass:retained-class}, for $B,Q_0\in\mathcal H$ and nonnegative radii,
\begin{equation}
 \Phi(B,\delta_B;Q_0)\le\delta_B+W_1(B,Q_0)
 \label{rec:triangle-upper}
\end{equation}
gives an inexpensive sufficient certificate.
In particular,
\begin{equation}
 \delta_B+W_1(B,Q_0)\le\delta_0
 \label{rec:nesting-condition}
\end{equation}
together with $\delta_B\ge r^\star(B)$ guarantees \eqref{rec:nested-dro} and, under Assumption~\ref{ass:moment-sensitivity},
\begin{equation}
 \Theta_{\mathrm{ID}}(P_X)\subseteq\mathcal T_{\delta_B}(B)
 \subseteq\mathcal T_{\delta_0}(Q_0).
 \label{rec:nested-tubes}
\end{equation}
For $\delta_0\ge r^\star(Q_0)$, the triangle-certified benchmarks $\{B\in\mathcal H:r^\star(B)+W_1(B,Q_0)\le\delta_0\}$ form a nonempty compact convex set.
Minimizing $r^\star$ on this set is a convex program with an attained minimum, and an LP when the compatible class has finite generators on finite support with linear retained restrictions (Appendix~\ref{app:recenter-finite}).
\end{corollary}

The exact criterion guarantees weak DRO-set inclusion.
A smaller radius need not strictly tighten the parameter set, nor does exact envelope nesting imply nesting of the separately constructed tubes.
For a candidate accepted only by the sharper test, report $\mathcal T_{\delta_0}(Q_0)\cap\mathcal T_{\delta_B}(B)$ when a nested tube refinement is desired.

\paragraph{Benchmark selection.}
Given a declared family $\mathcal B\subseteq\mathcal H$ and a valid coverage certificate $\overline r(B)\ge r^\star(B)$, the sharper ideal program is
\begin{equation}
 \min_{B\in\mathcal B}\overline r(B)
 \quad\text{subject to}\quad
 \Phi(B,\overline r(B);Q_0)\le\delta_0.
 \label{rec:center-rule}
\end{equation}
For the same family and radius rule, its feasible set contains every triangle-certified candidate, so its global optimum cannot have a larger radius and can have a strictly smaller one.
General convexity is not asserted.
Appendix~\ref{app:exact-nesting} gives conditions for attainment, and Appendix~\ref{app:three-point} illustrates a strict gain over the triangle rule.

The triangle rule remains useful: computable coverage and movement bounds $\overline r(\alpha)$ and $\overline s(\alpha)$ yield
\begin{equation}
 \min_{\alpha\in A}\{\overline r(\alpha):
             \overline r(\alpha)+\overline s(\alpha)\le\delta_0\}.
 \label{rec:certificate-rule}
\end{equation}
With compact convex $A$, continuous convex certificates, and a feasible baseline, \eqref{rec:certificate-rule} is a convex program with an attained minimum.
Its solution certifies coverage and nested tubes; exact nesting can also certify benchmarks that fail the triangle test.

Algorithm~\ref{alg:main-dro-construction} combines calibration, benchmark selection, and moment optimization into a single procedure.
It verifies the two proposals from the constrained and unrestricted radius programs before computing the final bounds.
Figure~\ref{fig:algorithm-recentering} summarizes the algorithm and the corresponding inclusions of law envelopes and parameter sets.

For a fixed proposal $(B,\delta_B)$, let $U_\Phi\ge\Phi(B,\delta_B;Q_0)$ denote a certified upper bound, including numerical error; use $U_\Phi=+\infty$ when no such certificate is available.
The nesting test evaluates the proposed pair without changing it.

\par\addvspace{\baselineskip}
\begingroup
\setstretch{1.05}
\Needspace{\textheight}
\hrule\vspace{2pt}
\captionsetup{font=normalsize,labelfont=bf,labelsep=space,
  justification=raggedright,singlelinecheck=false,skip=2pt,hypcap=false}
\captionof{algorithm}{Constructing and certifying recentered bounds}
\label{alg:main-dro-construction}
\hrule\vspace{3pt}
\textbf{Inputs.}
Measurement map $f:\X^\star\to\X$, maintained class $\mathcal C$, observed law $P_X$, moment $g_\theta$, parameter space $\Theta$, metric $d$, and retained class $\mathcal H$ satisfying Assumptions~\ref{ass:ambient-regularity}--\ref{ass:retained-class}; benchmarks $Q_\alpha\in\mathcal H$, $\alpha\in A$, and certified moment routines.
For tubes, also supply $L(\theta)$ as in Assumption~\ref{ass:moment-sensitivity}.
\begin{enumerate}[label=\textbf{\arabic*.},leftmargin=*,itemsep=4pt]
\item \textbf{Choose an initial benchmark and a radius ensuring coverage.}
Choose $Q_0\in\mathcal H$; calibrate by \eqref{eq:calibration-certificate} or \eqref{cal:covering-bound} and initialize
\[
 \delta_0\ge r^\star(Q_0),\qquad
 (Q_{\mathrm c},\delta_{\mathrm c})\gets(Q_0,\delta_0).
\]
The current pair $(Q_{\mathrm c},\delta_{\mathrm c})$ remains available if no proposal improves it.
\item \textbf{Propose new benchmarks and their covering radii.}
Bound each benchmark's covering radius and distance from $Q_0$:
\[
 r^\star(Q_\alpha)\le\overline r(\alpha),\qquad
 W_1(Q_\alpha,Q_0)\le\overline s(\alpha).
\]
Under the conditions of \eqref{rec:certificate-rule}, compute
\begin{equation}
 \begin{aligned}
 \alpha_\triangle&\in\operatorname*{arg\,min}_{\alpha\in A}\{\overline r(\alpha):\overline r(\alpha)+\overline s(\alpha)\le\delta_0\},\\
 \alpha_{\rm free}&\in\operatorname*{arg\,min}_{\alpha\in A}\overline r(\alpha).
 \end{aligned}
 \label{rec:unrestricted-proposal}
\end{equation}
The first program enforces nesting; the second seeks a smaller radius without imposing nesting, which Step~3 must verify.\footnote{%
In the finite-support setting of Appendix~\ref{app:recenter-finite}, solve \eqref{app:finite-center-lp} to obtain the unrestricted proposal; adding \eqref{app:finite-center-nest-lp} gives the triangle-certified proposal.}

For each $\alpha\in\{\alpha_\triangle,\alpha_{\rm free}\}$, propose the benchmark $B=Q_\alpha$ with covering radius $\delta_B=\overline r(\alpha)$.

\item \textbf{Check nesting and keep the smallest certified radius} (Theorem~\ref{main:recentering}, Corollary~\ref{main:triangle-certificate}).
Coverage follows from the Step~2 certificate $\delta_B=\overline r(\alpha)\ge r^\star(B)$.
Check nesting relative to the initial pair:
\begin{equation}
 \min\{\delta_B+\overline s(\alpha),U_\Phi\}\le\delta_0.
 \label{rec:acceptance-test}
\end{equation}
\begin{algorithmic}
\For{each $\alpha\in\{\alpha_\triangle,\alpha_{\rm free}\}$ and its pair $(B,\delta_B)$ from Step~2}
 \State $U_\Phi\gets+\infty$; first check the triangle certificate $\delta_B+\overline s(\alpha)\le\delta_0$.%
 \If{the triangle certificate does not establish nesting ($\delta_B+\overline s(\alpha)>\delta_0$)}
  \State Compute a sharper certificate $U_\Phi\ge\Phi(B,\delta_B;Q_0)$ analytically or using:\footnote{%
For deterministic network imputations satisfying the restrictions in Appendix~\ref{app:network-nesting}, use the exact LP \eqref{app:eq:network-nesting-lp}.}
  \Statex \hspace{\algorithmicindent}\hspace{\algorithmicindent}%
  (i) point-mass $Q_0$: moment program \eqref{eq:pointmass-phi};
  \Statex \hspace{\algorithmicindent}\hspace{\algorithmicindent}%
  (ii) finite support, linear restrictions: upper relaxation \eqref{eq:mccormick-phi}.
  \State Keep $U_\Phi=+\infty$ if no upper certificate is obtained.
 \EndIf
 \If{\eqref{rec:acceptance-test} holds \textbf{and} $\delta_B<\delta_{\mathrm c}$}
  \State $(Q_{\mathrm c},\delta_{\mathrm c})\gets(B,\delta_B)$
 \Else
  \State Keep $(Q_{\mathrm c},\delta_{\mathrm c})$ unchanged.
 \EndIf
\EndFor
\end{algorithmic}
\end{enumerate}
\par\vspace{3pt}\hrule

\newpage
\hrule\vspace{2pt}
\noindent\textbf{Algorithm~\ref*{alg:main-dro-construction}} (continued)\par
\vspace{2pt}\hrule\vspace{3pt}
\begin{enumerate}[label=\textbf{\arabic*.},leftmargin=*,itemsep=4pt,start=4]
\item \textbf{Compute the parameter set for the selected pair} (Theorem~\ref{main:feasibility}).
Minimize and maximize the moment over the selected envelope; retain $\theta$ when zero lies between these bounds, as in \eqref{main:feasibility-test}:
\[
 \Theta_{\mathrm{DRO}}(\delta_{\mathrm c};Q_{\mathrm c})
 =\left\{\theta\in\Theta:
 \min_{Q\in\mathcal A_{\delta_{\mathrm c}}(Q_{\mathrm c})}\E_Qg_\theta
 \le0\le
 \max_{Q\in\mathcal A_{\delta_{\mathrm c}}(Q_{\mathrm c})}\E_Qg_\theta\right\}.
\]
\item \textbf{Compute both tubes and report their intersection} (Theorem~\ref{main:tube}).
Using the initial pair from Step~1 and the selected pair from Step~3, compute
\[
 \begin{aligned}
 \mathcal T_{\delta_0}(Q_0)
 &\gets\{\theta\in\Theta:|\E_{Q_0}g_\theta|\le\delta_0L(\theta)\},\\
 \mathcal T_{\delta_{\mathrm c}}(Q_{\mathrm c})
 &\gets\{\theta\in\Theta:|\E_{Q_{\mathrm c}}g_\theta|\le\delta_{\mathrm c}L(\theta)\},\\
 \mathcal T_{\rm out}
 &\gets\mathcal T_{\delta_0}(Q_0)\cap\mathcal T_{\delta_{\mathrm c}}(Q_{\mathrm c}).
 \end{aligned}
\]
\end{enumerate}
\textbf{Output.}
Return $(Q_{\mathrm c},\delta_{\mathrm c})$, $\Theta_{\mathrm{DRO}}(\delta_{\mathrm c};Q_{\mathrm c})$, and, when computed, $\mathcal T_{\rm out}$.
Both sets contain $\Theta_{\mathrm{ID}}(P_X)$; the DRO set refines the initial set by \eqref{rec:nested-dro}.
\par\vspace{3pt}\hrule
\endgroup
\par\addvspace{\baselineskip}

Choose among Step~3's methods according to the problem's structure.
Exact verification on small supports is also available in Appendix~\ref{app:finite-nesting}.
An upper bound exceeding $\delta_0$ alone is inconclusive; a feasible envelope law with distance from $Q_0$ exceeding $\delta_0$ proves failure of envelope nesting.
The regression application in Appendix~\ref{app:recentering-eiv} illustrates exact acceptance when the triangle bound is inconclusive.
The strict improvement rule retains the current pair on ties.
No global solution of \eqref{rec:center-rule} is required.

Step~4 can use affine endpoint search \eqref{app:affine-computation-endpoints} or certified set inversion under the conditions of Proposition~\ref{app:parameter-computation-guarantee}; at finite tolerances, report the certified enclosures on the chosen search interval.
In Step~5, the intersection equals the new tube whenever the triangle certificate guarantees tube nesting.
A benchmark-specific bound on the moment change can further tighten Step~5's tubes, as in Appendix~\ref{app:recentering-eiv}.
Alternatively, intersect calibrated \emph{law} envelopes and solve Step~4's moment programs over that intersection (see the intersection fallback in Appendix~\ref{app:recenter-finite}).

\begin{figure}[htbp]
\centering
\resizebox{\linewidth}{!}{\input{figures/recentering_algorithm_schematic}}
\caption{From calibrated law envelopes to parameter bounds.
Coverage preserves the compatible class in blue; exact nesting keeps the teal envelope inside the gray one.
The triangle condition is a sufficient special case, drawn here schematically.
Moment optimization maps these law inclusions to the parameter inclusions on the left.}
\label{fig:algorithm-recentering}\label{fig:ambiguity-geometry}
\end{figure}
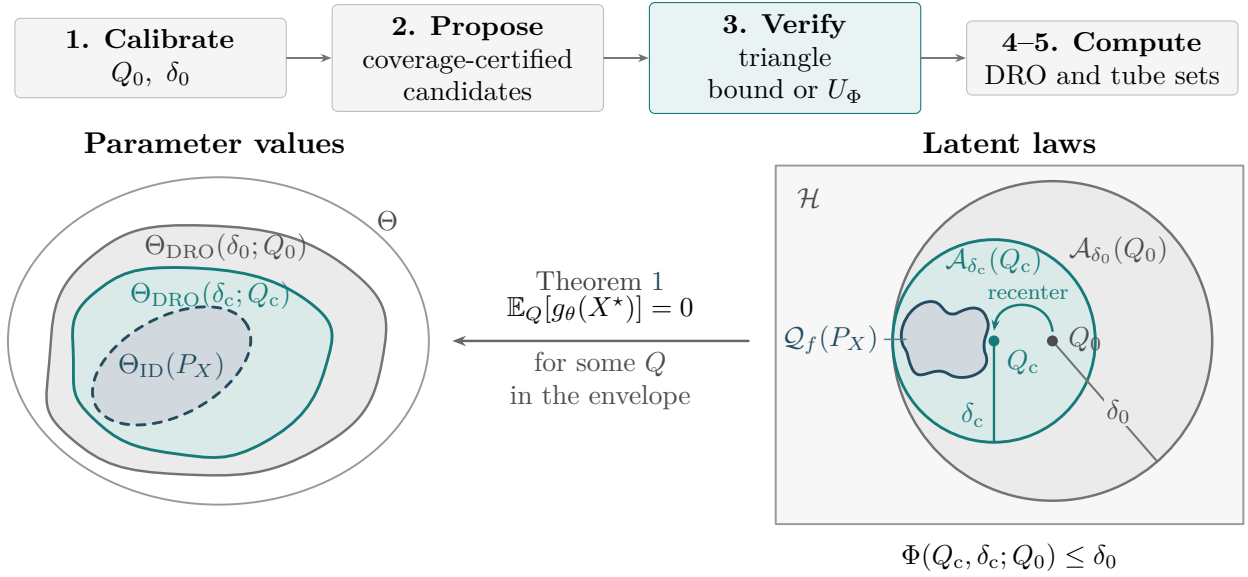
\FloatBarrier

\subsection{Recentering in the examples}
\label{sec:recentering-examples}
Building on Section~\ref{sec:dro-examples}, the examples apply Algorithm~\ref{alg:main-dro-construction} to choose a new benchmark and radius, then recompute the parameter bounds.

\paragraph{Missing outcomes.}
Here recentering removes all excess width and recovers the sharp identified interval.
Keep $g_\theta(D,Y)=\theta-Y$, the zero imputation $Q_0=Q_{t=0}$, and its radius $\delta_0=p_{00}$.
A candidate $Q_t$ leaves observed outcomes unchanged and assigns success probability $t$ to missing outcomes.
Its covering and movement certificates are exact:
\[
 \overline r(t)=r^\star(Q_t)=p_{00}\max\{t,1-t\},\qquad
 \overline s(t)=W_1(Q_t,Q_0)=p_{00}t.
\]
Equation~\eqref{rec:certificate-rule} becomes
\[
 \min_{0\le t\le1}
 \left\{p_{00}\max\{t,1-t\}:
        p_{00}\max\{t,1-t\}+p_{00}t\le p_{00}\right\}.
\]
The midpoint $t=1/2$ solves this problem, giving
\[
 (Q_{\mathrm c},\delta_{\mathrm c})=(Q_{1/2},p_{00}/2),\qquad
 \delta_{\mathrm c}+W_1(Q_{\mathrm c},Q_0)=p_{00}=\delta_0.
\]
This benchmark attains the globally smallest covering radius over $\mathcal H$, hence also solves \eqref{rec:center-rule} with exact radii.
The displayed equality verifies the nesting condition \eqref{rec:nesting-condition} in Corollary~\ref{main:triangle-certificate}.
For Theorem~\ref{main:feasibility}, the exact moment extrema are
\[
 \begin{aligned}
 \min_{Q\in\mathcal A_{\delta_{\mathrm c}}(Q_{\mathrm c})}
   \E_Q[\theta-Y]&=\theta-(p_{11}+p_{00}),\\
 \max_{Q\in\mathcal A_{\delta_{\mathrm c}}(Q_{\mathrm c})}
   \E_Q[\theta-Y]&=\theta-p_{11}.
 \end{aligned}
\]
They are attained by imputing one and zero for all missing outcomes, respectively.
Requiring zero to lie between these extrema gives
\[
 \Theta_{\mathrm{DRO}}(\delta_{\mathrm c};Q_{\mathrm c})
 =[p_{11},p_{11}+p_{00}].
\]
The inequality in Theorem~\ref{main:tube}, $|\theta-(p_{11}+p_{00}/2)|\le p_{00}/2$, gives the same interval for the tube.
In the numerical example, recentering raises the lower bound from $0.2$ to $0.4$, while the upper bound remains $0.6$.
Both the DRO interval and the tube therefore have half their original width.
Appendix~\ref{app:recentering-missing} provides details.

\paragraph{Thresholded networks.}
Here recentering tightens the computable bounds without requiring the sharp set to be computed.
Continue with the same observed law, capacity-weighted distance, and initial radius $\delta_0=\eta d_{\max}$.
Let $G_t(G)$ assign the fraction $t\in[0,1/2]$ of each eligible weak-link capacity and reduce its untreated-contact weight to keep the row sum equal to one.
Write $Q_t=\mathcal L_{P_X}(X,G_t(G))$.
The exact covering and movement certificates are
\[
 \overline r(t)=r^\star(Q_t)=(1-t)\delta_0,
 \qquad \overline s(t)=W_1(Q_t,Q_0)=t\delta_0.
\]
Consequently, \eqref{rec:certificate-rule} becomes
\[
 \min_{0\le t\le1/2}
 \bigl\{(1-t)\delta_0:
       (1-t)\delta_0+t\delta_0\le\delta_0\bigr\}.
\]
Its solution is the midpoint imputation, $Q_{\mathrm c}=Q_{1/2}$, with $\delta_{\mathrm c}=\delta_0/2$.
Corollary~\ref{main:triangle-certificate} applies because the new radius and movement sum to $\delta_0$.
Compatible laws separated by $\delta_0$ show that every benchmark covering them needs radius at least $\delta_0/2$.
Thus this midpoint is already globally optimal for covering radius; the sharper criterion cannot reduce that radius further.

The recentered exposure bounds are
\[
 \underline b_{\mathrm c}=1-\delta_0/2,\qquad
 \overline b_{\mathrm c}
 =1+\frac{\eta}{2}\sum_{(i,j)\in E}w_{ij}
     +\frac{\delta_0}{2}\sum_{i=1}^m\omega_i.
\]
The implementation computes these exposure extrema by constrained optimization, supplies them to the moment routine, and searches for the parameter endpoints.
The formulas validate those computations.
For the same positive observed contribution, the resulting interval is
\[
 \Theta_{\mathrm{DRO}}(\delta_{\mathrm c};Q_{\mathrm c})
 =\Theta\cap\left[
 \frac{\E_{P_X}[T^\top(Y-\beta_0T)]}{\overline b_{\mathrm c}},\,
 \frac{\E_{P_X}[T^\top(Y-\beta_0T)]}{\underline b_{\mathrm c}}
 \right].
\]
The positive floor in the metric weights yields $\delta_0\sum_{i=1}^m\omega_i-\eta\sum_e w_e =\eta(m d_{\max}-\sum_e w_e)/10>0$ in this instance.
Both interval endpoints move inward when $\delta_0>0$ and the observed contribution is positive, before intersection with $\Theta$.
Thus recentering strictly reduces the width of this unrestricted interval.
It need not recover the sharp set.

Holding the same observed contribution fixed, Theorem~\ref{main:tube} gives the nested exposure tube
\[
 \left|\E_{P_X}[T^\top(Y-\beta_0T)]
  -\theta\left(1+\frac{\eta}{2}\sum_{(i,j)\in E}w_{ij}\right)\right|
 \le\frac{\delta_0}{2}\left(1+\sum_{i=1}^m\omega_i\right)|\theta|.
\]
Appendix~\ref{app:recentering-network} proves these calculations and gives the numerical bounds in Figure~\ref{fig:recentering-sets}.

\begin{remark}
Known link capacities can further strengthen $\mathcal H$ within the same moment-optimization framework.
In this example, retaining the capacity relaxation gives tighter bounds than the displayed DRO intervals; both calibrated transport balls then contain the entire retained class, so recentering leaves that bound unchanged.
Inexpensive maximum-cut certificates can strengthen it further (Appendix~\ref{app:network-comparators}).
\end{remark}

\begin{figure}[htbp]
\centering
\input{figures/bounds_before_after_recentering}
\caption{Population bounds before and after recentering.
The network intervals bound the moment-defined exposure slope in \eqref{eq:network-moment}.
The network tubes hold the observed contribution fixed before applying Theorem~\ref{main:tube}.
The missing-outcome sharp interval is known exactly.
For networks, the solid sharp segment is a certified inner interval and the dashed segment brackets the unknown lower endpoint.}
\label{fig:recentering-sets}\label{fig:dro-example-sets}
\end{figure}
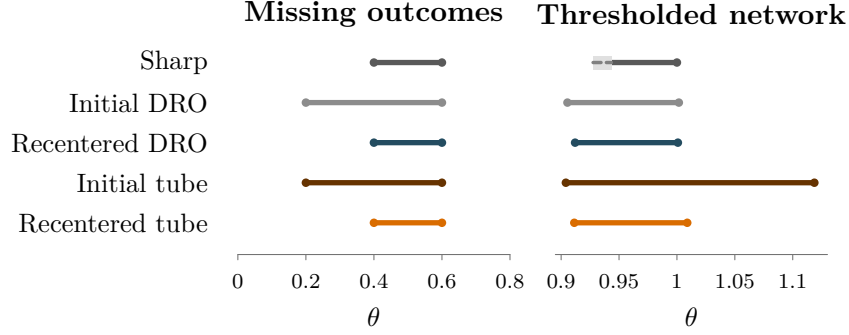
\FloatBarrier

Recentering can remove excess width from an outer approximation, but it cannot reduce the sharp identified set without additional information or assumptions.

\input{sections/unobserved_component_bounds}

%% file: figures/recentering_algorithm_schematic.tex
\begin{tikzpicture}[
    x=1cm,y=1cm,font=\footnotesize,
    line cap=round,line join=round,
    every node/.style={inner sep=2pt},
    stepbox/.style={draw=black!25,fill=black!4,rounded corners=2pt,
      text width=3.2cm,minimum height=.95cm,align=center,inner sep=4pt}
]
\definecolor{algorithmteal}{HTML}{157A78}

\node[stepbox] (calibrate) at (1.78,6.08)
  {\textbf{1. Calibrate}\\$Q_0,\ \delta_0$};
\node[stepbox] (certify) at (5.85,6.08)
  {\textbf{2. Propose}\\coverage-certified candidates};
\node[stepbox,draw=algorithmteal,fill=algorithmteal!9]
  (recenter) at (9.92,6.08)
  {\textbf{3. Verify}\\triangle bound or $U_\Phi$};
\node[stepbox] (compute) at (13.99,6.08)
  {\textbf{4--5. Compute}\\DRO and tube sets};
\foreach \from/\to in {calibrate/certify,certify/recenter,recenter/compute}
  \draw[black!55,-{Stealth[length=4pt]},line width=.7pt]
    (\from.east) -- (\to.west);

\node[font=\small\bfseries] at (2.6,5.0) {Parameter values};
\draw[black!40,line width=.7pt]
  (2.65,2.45) ellipse [x radius=2.70,y radius=2.10];
\node[text=black!70] at (4.82,4.00) {$\Theta$};
\path[fill=black!8,draw=black!55,line width=.9pt]
  plot[smooth cycle,tension=.85] coordinates {
    (.48,2.18) (.96,3.48) (2.50,3.93) (4.32,3.38)
    (4.75,2.19) (3.89,1.05) (2.15,.74) (.76,1.20)
  };
\node[text=black!70] at (2.75,3.65)
  {$\Theta_{\mathrm{DRO}}(\delta_0;Q_0)$};
\path[fill=algorithmteal!16,draw=algorithmteal,line width=1.05pt]
  plot[smooth cycle,tension=.85] coordinates {
    (.79,2.20) (1.04,3.05) (2.22,3.39) (3.77,2.98)
    (4.02,2.10) (3.24,1.15) (2.05,1.01) (1.05,1.45)
  };
\node[text=algorithmteal] at (2.55,3.04)
  {$\Theta_{\mathrm{DRO}}(\delta_{\mathrm c};Q_{\mathrm c})$};
\begin{scope}[shift={(2.05,2.13)},rotate=28]
  \path[fill=inkblue!21,draw=inkblue,dashed,line width=1.05pt]
    (0,0) ellipse [x radius=1.10,y radius=.62857];
\end{scope}
\node[text=inkblue] at (2.05,2.13) {$\Theta_{\mathrm{ID}}(P_X)$};

\draw[black!60,{Stealth[length=5pt]}-,line width=.9pt]
  (5.65,2.45) -- (9.45,2.45);
\node[text=black!75] at (7.55,3.22) {Theorem~\ref{main:feasibility}};
\node at (7.55,2.85) {$\E_Q[g_\theta(X^\star)]=0$};
\node[align=center,text=black!70] at (7.55,1.94)
  {for some $Q$\\in the envelope};

\node[font=\small\bfseries] at (12.80,5.0) {Latent laws};
\filldraw[fill=black!3,draw=black!40,line width=.7pt]
  (9.8,.10) rectangle (15.8,4.70);
\node[anchor=north west,text=black!70] at (10.0,4.48) {$\mathcal H$};
\path[fill=black!8,draw=black!55,line width=.9pt]
  (13.35,2.45) circle (2.05);
\node[text=black!70] at (14.20,3.62) {$\mathcal A_{\delta_0}(Q_0)$};
\path[fill=algorithmteal!16,draw=algorithmteal,line width=1.05pt]
  (12.60,2.45) circle (1.30);
\node[text=algorithmteal] at (12.60,3.48)
  {$\mathcal A_{\delta_{\mathrm c}}(Q_{\mathrm c})$};

\begin{scope}[shift={(12.00,2.45)},scale=.70]
  \input{figures/compatible_class_outline}
\end{scope}
\node[text=inkblue] (compatiblelabel) at (10.50,2.47) {$\Qset_f(P_X)$};
\draw[inkblue!75,line width=.55pt]
  (compatiblelabel.east) -- (11.47,2.47);

\draw[black!55,line width=.75pt] (13.35,2.45) -- (14.695, .903);
\node[fill=black!8,text=black!70,inner sep=1pt]
  at (14.20,1.54) {$\delta_0$};
\draw[algorithmteal,line width=.85pt] (12.60,2.45) -- (12.60,1.15);
\node[text=algorithmteal] at (12.33,1.50) {$\delta_{\mathrm c}$};
\draw[algorithmteal,-{Stealth[length=4pt]},line width=.9pt]
  (13.35,2.58) .. controls (13.30,3.01) and (12.67,3.01) .. (12.60,2.58);
\node[font=\scriptsize,text=algorithmteal] at (13.04,3.08) {recenter};
\fill[black!70] (13.35,2.45) circle (2pt);
\node[anchor=west,text=black!70] at (13.48,2.47) {$Q_0$};
\fill[algorithmteal] (12.60,2.45) circle (2pt);
\node[anchor=west,text=algorithmteal] at (12.69,2.18) {$Q_{\mathrm c}$};

\node at (12.80,-.32)
  {$\Phi(Q_{\mathrm c},\delta_{\mathrm c};Q_0)\le\delta_0$};
\end{tikzpicture}

%% file: figures/bounds_before_after_recentering.tex
\begin{tikzpicture}[x=1cm,y=1cm,font=\footnotesize,line cap=round]
\node[font=\small\bfseries] at (5.3,.65) {Missing outcomes};
\node[font=\small\bfseries] at (9.5,.65) {Thresholded network};
\node[anchor=east] at (3.25,-0.000) {Sharp};
\draw[black!65,line width=2pt] (5.300000,-0.000) -- (6.200000,-0.000);
\fill[black!65] (6.200000,-0.000) circle (1.6pt);
\fill[black!65] (5.300000,-0.000) circle (1.6pt);
\draw[black!65,line width=2pt] (8.441507,-0.000) -- (9.308511,-0.000);
\fill[black!65] (9.308511,-0.000) circle (1.6pt);
\fill[black!12] (8.199103,-0.090) rectangle (8.441507,0.090);
\draw[black!50,densely dashed,line width=1.3pt] (8.199103,-0.000) -- (8.441507,-0.000);
\node[anchor=east] at (3.25,-0.530) {Initial DRO};
\draw[black!45,line width=2pt] (4.400000,-0.530) -- (6.200000,-0.530);
\fill[black!45] (6.200000,-0.530) circle (1.6pt);
\fill[black!45] (4.400000,-0.530) circle (1.6pt);
\draw[black!45,line width=2pt] (7.860001,-0.530) -- (9.333916,-0.530);
\fill[black!45] (9.333916,-0.530) circle (1.6pt);
\fill[black!45] (7.860001,-0.530) circle (1.6pt);
\node[anchor=east] at (3.25,-1.060) {Recentered DRO};
\draw[inkblue,line width=2pt] (5.300000,-1.060) -- (6.200000,-1.060);
\fill[inkblue] (6.200000,-1.060) circle (1.6pt);
\fill[inkblue] (5.300000,-1.060) circle (1.6pt);
\draw[inkblue,line width=2pt] (7.961147,-1.060) -- (9.321203,-1.060);
\fill[inkblue] (9.321203,-1.060) circle (1.6pt);
\fill[inkblue] (7.961147,-1.060) circle (1.6pt);
\node[anchor=east] at (3.25,-1.590) {Initial tube};
\draw[orange!40!black,line width=2pt] (4.400000,-1.590) -- (6.200000,-1.590);
\fill[orange!40!black] (6.200000,-1.590) circle (1.6pt);
\fill[orange!40!black] (4.400000,-1.590) circle (1.6pt);
\draw[orange!40!black,line width=2pt] (7.839239,-1.590) -- (11.126514,-1.590);
\fill[orange!40!black] (11.126514,-1.590) circle (1.6pt);
\fill[orange!40!black] (7.839239,-1.590) circle (1.6pt);
\node[anchor=east] at (3.25,-2.120) {Recentered tube};
\draw[orange!85!black,line width=2pt] (5.300000,-2.120) -- (6.200000,-2.120);
\fill[orange!85!black] (6.200000,-2.120) circle (1.6pt);
\fill[orange!85!black] (5.300000,-2.120) circle (1.6pt);
\draw[orange!85!black,line width=2pt] (7.950606,-2.120) -- (9.444876,-2.120);
\fill[orange!85!black] (9.444876,-2.120) circle (1.6pt);
\fill[orange!85!black] (7.950606,-2.120) circle (1.6pt);
\draw[black!55] (3.5,-2.54) -- (7.1,-2.54);
\draw[black!55] (3.500000,-2.54) -- (3.500000,-2.61);
\node[anchor=north,font=\scriptsize] at (3.500000,-2.66) {0};
\draw[black!55] (4.400000,-2.54) -- (4.400000,-2.61);
\node[anchor=north,font=\scriptsize] at (4.400000,-2.66) {0.2};
\draw[black!55] (5.300000,-2.54) -- (5.300000,-2.61);
\node[anchor=north,font=\scriptsize] at (5.300000,-2.66) {0.4};
\draw[black!55] (6.200000,-2.54) -- (6.200000,-2.61);
\node[anchor=north,font=\scriptsize] at (6.200000,-2.66) {0.6};
\draw[black!55] (7.100000,-2.54) -- (7.100000,-2.61);
\node[anchor=north,font=\scriptsize] at (7.100000,-2.66) {0.8};
\draw[black!55] (7.7,-2.54) -- (11.3,-2.54);
\draw[black!55] (7.776596,-2.54) -- (7.776596,-2.61);
\node[anchor=north,font=\scriptsize] at (7.776596,-2.66) {0.9};
\draw[black!55] (8.542553,-2.54) -- (8.542553,-2.61);
\node[anchor=north,font=\scriptsize] at (8.542553,-2.66) {0.95};
\draw[black!55] (9.308511,-2.54) -- (9.308511,-2.61);
\node[anchor=north,font=\scriptsize] at (9.308511,-2.66) {1};
\draw[black!55] (10.074468,-2.54) -- (10.074468,-2.61);
\node[anchor=north,font=\scriptsize] at (10.074468,-2.66) {1.05};
\draw[black!55] (10.840426,-2.54) -- (10.840426,-2.61);
\node[anchor=north,font=\scriptsize] at (10.840426,-2.66) {1.1};
\node at (5.3,-3.36) {$\theta$};
\node at (9.5,-3.36) {$\theta$};
\end{tikzpicture}

%% file: sections/unobserved_component_bounds.tex
\paragraph{Bounding the unobserved component.}
\label{disc:benchmark-residual}
When part of a moment is determined by the observed law, we can hold its expectation fixed and bound only the remaining residual.
Step~4 of Algorithm~\ref{alg:main-dro-construction} can optimize the residual over its calibrated envelope, as in the regression application.
Further bounds from the maintained restrictions can also be intersected with the DRO set while preserving every compatible value.
Appendix~\ref{app:population-residuals} proves the residual identity; Appendices~\ref{app:network} and~\ref{app:measurement-error} give its applications.

%% file: sections/discussion_and_conclusion.tex
\section{Discussion and conclusion}
\label{sec:conclusion}

We asked how to compute and refine parameter bounds from imperfect measurements while preserving every value allowed by the data and assumptions.
Our framework combines measurement restrictions, moment optimization, and benchmark selection to refine outer bounds while holding the observed law and identifying assumptions fixed.

Algorithm~\ref{alg:main-dro-construction} implements our main contribution: a reusable methodology for calibration, benchmark selection, and parameter-bound computation.
The theory ensures that accepted benchmark changes preserve every compatible parameter value without widening the initial bounds.

The network application shows how known weight limits and benchmark selection yield tighter bounds on the moment-defined network-exposure slope without solving the difficult sharp identification problem or strengthening assumptions.
Recentering recovers the sharp interval for missing outcomes.
In regression, retaining the observed joint law and measurement-error budget makes the optimized bounds sharp.

Future work can develop methods for estimating the bounds and constructing confidence sets for the parameter and its identified set, accounting for sampling uncertainty and benchmark selection.
It can also clarify how imperfect measurement affects attainable estimation rates and the precision of inference.

%% file: appendix/supplementary_contents.tex
\begingroup
\setstretch{1}
\pdfbookmark[0]{Supplementary Material}{appendix-outline}
\begin{center}
{\Large Supplementary Material to\\[.2em]
``Partial Identification under Imperfect Measurement:\\
A Distributionally Robust Approach''\par}
\vspace{1.5em}
Isaac Meza\\[.3em]
Department of Economics, Harvard University
\end{center}

\vspace{2em}
\noindent{\large\bfseries Table of contents\par}
\vspace{5pt}
\small
\newcommand{\appendixoutlineentry}[2]{%
  \par\addvspace{8pt}\noindent
  \hyperref[#1]{\ref*{#1}\enspace #2}\hfill
  \hyperref[#1]{\textcolor{black}{\bfseries\pageref*{#1}}}\par
}
\appendixoutlineentry{app:main-proofs}{Foundations and core proofs}
\appendixoutlineentry{app:radius-calibration}{Coverage and benchmark refinement}
\appendixoutlineentry{app:computational-analysis}{Computing parameter sets}
\appendixoutlineentry{app:missing-outcomes}{Missing outcomes}
\appendixoutlineentry{app:network}{Thresholded networks}
\appendixoutlineentry{app:measurement-error}{Measurement-error regression}
\par\addvspace{8pt}\noindent
\hyperref[app:references]{References}\hfill
\hyperref[app:references]{\textcolor{black}{\bfseries\pageref*{app:references}}}
\par
\endgroup

%% file: appendix/foundations_and_core_proofs.tex
\section{Foundations and core proofs}
\label{app:main-proofs}
Throughout, $\X^\star$ is compact, $\mathcal H\subseteq\mathcal P(\X^\star)$ is nonempty, closed and convex, and benchmarks belong to $\mathcal H$.
The compatible class $\mathcal F=\Qset_f(P_X)$ is nonempty; it need not be closed or convex.
On $\X^\star$, $W_1$ metrizes weak convergence.

\subsection{Accommodating random measurement error}
\label{app:stochastic}\label{app:stoch:known-kernel}
Including measurement error in the latent state makes measurement deterministic and preserves the identified set and error restrictions.
Here $X^\star$ denotes the latent variables before adding $\epsilon$, so $X=f(X^\star,\epsilon)$.
Let $f$ be measurable between standard Borel spaces and $\nu(de\mid x^\star)$ the known law of $\epsilon$ given $X^\star=x^\star$.
For measurable $A\subseteq\X$, let $\mathbf1_A$ be its indicator.
Then
\[
 \mathsf K_f(A\mid x^\star)=\int\mathbf1_A\{f(x^\star,e)\}\nu(de\mid x^\star),
 \qquad \mathcal Q_f(P_X)=\{Q\in\mathcal C:\mathsf K_fQ=P_X\}.
\]
The induced observed law satisfies $(\mathsf K_fQ)(A)=\int\mathsf K_f(A\mid x^\star)Q(dx^\star)$; $\mathcal C$ imposes the other maintained restrictions.
\label{app:stoch:kernel-equivalence}
Indeed, conditional expectation proves necessity, and $Q(dx^\star)\nu(de\mid x^\star)\delta_{f(x^\star,e)}(dx)$ constructs a joint witness for sufficiency, where $\delta_z$ is the point mass at $z$.
A directly specified observation kernel has the same interpretation.
Deterministic measurement is the special case $\mathsf K_f(A\mid x^\star)=\mathbf1_A\{f(x^\star)\}$.

\label{app:stoch:partial-noise}\label{app:stoch:augmentation}
If the noise law is only partially specified, let $\mathcal R$ be the class of joint laws of $(X^\star,\epsilon)$ satisfying all maintained restrictions, including error-moment and dependence restrictions.
Write $R_{X^\star}$ for the marginal law of $X^\star$ under $R$.
Define
\[
 \mathcal R_f(P_X)=\{R\in\mathcal R:f_\#R=P_X\},\qquad
 \mathcal Q_f(P_X)=\{R_{X^\star}:R\in\mathcal R_f(P_X)\}.
\]
Augmenting the latent state by $\epsilon$ makes measurement deterministic.
It preserves the identified set because $\int g_\theta(x^\star)\,dR=\int g_\theta\,dR_{X^\star}$, and every projected compatible law has such a joint witness.
Projection does not remove the restrictions defining $\mathcal R$.
A transport relaxation on the projected space needs its own metric and coverage certificate, as in Appendix~\ref{app:measurement-error}.

\subsection{Proofs of Theorems~\ref*{main:feasibility} and~\ref*{main:tube}}
\label{app:feasibility-proof}
\label{app:duality-proof}
\begin{proof}[Proof of Theorem~\ref{main:feasibility}]
Since $\X^\star$ is compact, $\mathcal P(\X^\star)$ is weakly compact, as is its closed subset $\mathcal H$.
Continuity of $W_1$ then makes $\mathcal A_\delta(Q_0)$ compact; it is nonempty because it contains $Q_0$.
To see convexity, let $Q_1,Q_2$ belong to the envelope and let $\pi_1,\pi_2$ be optimal couplings to $Q_0$.
For $t\in[0,1]$, $t\pi_1+(1-t)\pi_2$ couples $tQ_1+(1-t)Q_2$ to $Q_0$, so
\[
 W_1(tQ_1+(1-t)Q_2,Q_0)
 \le tW_1(Q_1,Q_0)+(1-t)W_1(Q_2,Q_0)\le\delta.
\]
Convexity of $\mathcal H$ keeps this mixture admissible.

For continuous $g_\theta$, the map $Q\mapsto\psi_\theta(Q)=\int g_\theta\,dQ$ is continuous and affine.
It therefore attains a minimum $m_\theta$ and maximum $M_\theta$ on the envelope.
Mixtures of minimizing and maximizing laws attain every value in $[m_\theta,M_\theta]$, and no law has a moment outside that interval.
In particular, a zero-moment law exists exactly when $m_\theta\le0\le M_\theta$.

The set $\{Q\in\mathcal H:\psi_\theta(Q)=0\}$ is closed in compact $\mathcal H$.
When nonempty, it contains a law minimizing $W_1(Q,Q_0)$, so $R(\theta;Q_0)\le\delta$ is equivalent to a zero-moment law in the envelope, including at equality.
When it is empty, $R=+\infty$ and neither condition holds.
This proves \eqref{main:feasibility-test}.
Finally, $\delta\ge r^\star(Q_0)$ puts every $Q\in\mathcal F$ inside the envelope.
Each compatible zero-moment witness remains available, proving identification inclusion.
If $\mathcal H=\mathcal F$, coverage gives $\mathcal A_\delta(Q_0)=\mathcal F$, hence equality of the parameter sets.
\end{proof}

\begin{proof}[Proof of Theorem~\ref{main:tube}]
For every coupling $\pi$ of $Q,Q_0$,
\[
 |\psi_\theta(Q)-\psi_\theta(Q_0)|
 \le L(\theta)\int d(x,z)\,d\pi.
\]
Minimizing coupling cost gives the tube, also when $L(\theta)=0$.
The full symmetric moment interval is attained only if both directed moment deviations equal $\delta L(\theta)$.

For the dual, assume $\mathcal H=\mathcal P(\X^\star)$ and write $g=g_\theta$.
Let $\mathfrak P$ be the compact convex set of laws on $\X^\star\times\X^\star$ with second marginal $Q_0$.
Optimal couplings give
\[
 \max_{\substack{\pi\in\mathfrak P\\\int d\,d\pi\le\delta}}\int g(x)\,d\pi
 =\max_{\pi\in\mathfrak P}\inf_{\lambda\ge0}
   \left\{\lambda\delta+\int[g(x)-\lambda d(x,z)]\,d\pi\right\}.
\]
Sion's minimax theorem~\citeyearpar{Sion1958} interchanges the maximum and infimum.
For fixed $\lambda$, the maximum equals $\int\max_x\{g(x)-\lambda d(x,z)\}\,Q_0(dz)$.
To justify this equality without a selection assumption, maximize over a finite net, choose the first maximizing index measurably, and pass to the limit using uniform continuity on $\X^\star\times\X^\star$.
This proves \eqref{main:dual}, including $\delta=0$; apply it to $-g$ for the lower moment.
It is the compact-space specialization of Wasserstein duality in \citet{BlanchetMurthy2019,GaoKleywegt2023}.
For $\delta>0$ the dual objective is continuous and at least $\lambda\delta+\int g\,dQ_0$, hence attains its minimum.
At zero radius, a Lipschitz $g$ permits any $\lambda\ge\Lip(g)$, where $\Lip(g)$ is its least Lipschitz constant; merely continuous $g$ need not permit finite attainment (take $g(x)=\sqrt{x}$, $\X^\star=[0,1]$, $Q_0=\delta_0$).
For retained restrictions, the coupling's first marginal must remain in $\mathcal H$; the unrestricted formula cannot omit them.
\end{proof}

\subsection{Proof of Theorem~\ref*{main:recentering} and Corollary~\ref*{main:triangle-certificate}}
\label{app:recentering-main-proof}
\begin{proof}[Proof of Theorem~\ref{main:recentering}]
For a fixed $Q$ and $B_t=tB+(1-t)B'$ with $t\in[0,1]$, mixing optimal couplings from $Q$ to $B$ and $B'$ gives
\[
 W_1(Q,B_t)\le tW_1(Q,B)+(1-t)W_1(Q,B').
\]
Taking the supremum over $Q\in\mathcal F$ proves $r^\star(B_t)\le t r^\star(B)+(1-t)r^\star(B')$.
Also, $W_1(Q,B)\le W_1(Q,B')+W_1(B',B)$ for every compatible $Q$.
Taking suprema and interchanging $B,B'$ give
\begin{equation}
 |r^\star(B)-r^\star(B')|\le W_1(B,B').
 \label{app:rec:radius-lipschitz}
\end{equation}
Thus $r^\star$ is continuous and attains a minimum on compact $\mathcal H$.
These arguments do not require a maximizing law in $\mathcal F$, which need not itself be compact.

The new envelope is nonempty and compact, and $Q\mapsto W_1(Q,Q_0)$ is continuous, so the maximum defining $\Phi$ is attained.
Its value is at most $\delta_0$ exactly when every $Q\in\mathcal A_{\delta_B}(B)$ satisfies $W_1(Q,Q_0)\le\delta_0$.
All such laws already belong to $\mathcal H$, so this is precisely the envelope inclusion in \eqref{rec:exact-nesting}.
If $\delta_B\ge r^\star(B)$, then
\[
 \mathcal F\subseteq\mathcal A_{\delta_B}(B)
 \subseteq\mathcal A_{\delta_0}(Q_0).
\]
A law witnessing a zero moment in either smaller class remains a witness in the larger class.
This proves both inclusions in \eqref{rec:nested-dro}.
\end{proof}

\begin{proof}[Proof of Corollary~\ref{main:triangle-certificate}]
For every law in the new envelope, the triangle inequality gives
\[
 W_1(Q,Q_0)\le W_1(Q,B)+W_1(B,Q_0)
 \le\delta_B+W_1(B,Q_0),
\]
which proves \eqref{rec:triangle-upper}.
Condition \eqref{rec:nesting-condition} therefore implies exact nesting.
For $\theta\in\mathcal T_{\delta_B}(B)$, the Lipschitz moment bound from Theorem~\ref{main:tube} yields
\[
 \begin{aligned}
 |\psi_\theta(Q_0)|
 &\le |\psi_\theta(B)|+L(\theta)W_1(B,Q_0)\\
 &\le L(\theta)\{\delta_B+W_1(B,Q_0)\}
 \le L(\theta)\delta_0.
 \end{aligned}
\]
Hence the tubes are nested, also when $L(\theta)=0$.
Coverage and Theorem~\ref{main:tube} give the identification inclusion in \eqref{rec:nested-tubes}.

Both $r^\star(B)$ and $W_1(B,Q_0)$ are continuous and convex in $B$.
The $\delta_0$-sublevel set of their sum in $\mathcal H$ is therefore compact and convex.
It contains $Q_0$ when $\delta_0\ge r^\star(Q_0)$, so minimizing the continuous convex function $r^\star$ there is an attained convex program.
On finite support, each Wasserstein distance has a transport LP.
With finitely many compatible generators, convexity reduces the coverage supremum to the maximum of their distances.
The joint epigraph formulation with linear retained restrictions is thus an LP, written explicitly in Appendix~\ref{app:recenter-finite}.
\end{proof}

\subsection{Observable contributions and residual bounds}
\label{app:population-residuals}\label{app:residual-proof}
Section~\ref{disc:benchmark-residual} proposes holding observed moment contributions fixed and bounding only the residual.
The result below justifies the resulting outer parameter bounds and, under Lipschitz regularity, the network example's exposure tube.

Choose an integrable observable function $h_\theta$ and set
\begin{equation}
 e_\theta=g_\theta-h_\theta\circ f,\qquad
 \widetilde g_\theta=\E_{P_X}h_\theta+e_\theta.
 \label{eq:pop-residual-center}
\end{equation}
\begin{proposition}[Residual identity and bounds]\label{pop:residual}
If $\mathcal U\supseteq\mathcal F$ and the residual is integrable under every $Q\in\mathcal U$, a valid outer parameter set is
\begin{equation}
 \mathcal R_h(\mathcal U)=\left\{\theta\in\Theta:
 \E_{P_X}h_\theta+\inf_{Q\in\mathcal U}\E_Qe_\theta\le0\le
 \E_{P_X}h_\theta+\sup_{Q\in\mathcal U}\E_Qe_\theta\right\}.
 \label{eq:pop-residual-set}
\end{equation}
For $\mathcal U=\mathcal A_\delta(Q_0)$ and an $L_{\rm res}(\theta)$-Lipschitz residual, this set is contained in $\{\theta:|\E_{P_X}h_\theta+\E_{Q_0}e_\theta| \le\delta L_{\rm res}(\theta)\}$.
\end{proposition}
\begin{proof}
Compatibility gives $\E_Q(h_\theta\circ f)=\E_{P_X}h_\theta$, hence $\E_Q\widetilde g_\theta=\E_Qg_\theta$.
A compatible zero-moment witness proves the first claim; the coupling Lipschitz inequality proves the second.
For a known measurement kernel replace $h_\theta\circ f$ by $\mathsf K_f^*h_\theta(x^\star) =\int h_\theta(x)\mathsf K_f(dx\mid x^\star)$, with the same proof.
\end{proof}
The corrected and original moment optimizations agree whenever $\mathcal U$ retains the observable expectation.
Otherwise they need not be nested; intersect their valid outer sets.
With compact convex $\mathcal U$ and continuous residual, its moment interval is attained.

%% file: appendix/coverage_and_benchmark_refinement.tex
\clearpage
\section{Coverage and benchmark refinement}
\label{app:radius-calibration}\label{app:recentering}

\subsection{Coverage certificates}
\label{app:recenter-coupling}
\label{rc:coupling-certificate}\label{app:radius-certificate-lp}
Algorithm~\ref{alg:main-dro-construction} requires a radius that covers every compatible distribution.
This subsection shows how to obtain such a radius for benchmarks constructed by imputing from observed data, using the maintained measurement restrictions.
It provides both moment-based bounds and a finite-state linear program.

Let $Z$ be the full state, $f(Z)$ its observation, and $\tau(Z)$ the state on which transport is measured.
For compatible joint distributions $R\in\mathfrak F$ of the full state $Z$, $f_\#R=P_X$ and $\mathcal F=\{\tau_\#R:R\in\mathfrak F\}$.
An imputation kernel $K_B(du\mid x)$ defines $B(du)=\int K_B(du\mid x)P_X(dx)$.
Under the joint law $R(dz)K_B(du\mid f(z))$, the pair $(\tau(Z),U)$ has marginals $\tau_\#R$ and $B$, so
\begin{equation}
 r^\star(B)\le\sup_{R\in\mathfrak F}\E_R c_B(Z),\qquad
 c_B(z)=\int d(\tau(z),u)K_B(du\mid f(z)).
 \label{cal:covering-bound}
\end{equation}
This imposes no independence restriction among the components of $Z$.
For a deterministic imputation $b$, if $p\ge1$, $\lambda_j\ge0$ and
\[
 d(\tau(z),b(f(z)))^p\le h(f(z))+\sum_j\lambda_jm_j(z),
 \qquad \E_Rm_j\le M_j,
\]
then Jensen's inequality gives the certificate $r^\star(b_\#P_X)\le(\E_{P_X}h+\sum_j\lambda_jM_j)^{1/p}$.
All terms must be integrable and the inequality must hold on the maintained support.
Nonemptiness ensures a nonnegative radicand.
Thus a uniform bound $M$ on expected squared imputation distance gives a $W_1$ radius $\sqrt M$ without changing the transport metric to $W_2$.

When $Z$ takes finitely many values, write compatibility of a probability vector $q$ as $q\ge0$, $Aq=a$, $Cq\le m$, including normalization and the observed law, with known constraint matrices $A,C$ and vectors $a,m$.
For imputation costs $c$, the certificate is
\begin{equation}
 \max_{q\ge0,Aq=a,Cq\le m}c^\top q
 =\min_{A^\top v+C^\top\lambda\ge c,\lambda\ge0}
       (a^\top v+m^\top\lambda).
 \label{cal:finite-certificate-lp}
\end{equation}
Any dual-feasible solution gives an upper certificate; $v$ is unrestricted.
When the compatible polytope is explicitly available, sharp moment optimization is also an LP.
Computational savings therefore depend on whether a covering-radius certificate is easier to obtain than the required sharp moment bounds.

\subsection{Continuity and benchmark selection}
\label{app:operational-centering}\label{app:exact-nesting}
This subsection gives conditions under which the benchmark-selection programs in Section~\ref{sec:recentering} have an optimal solution.
It also explains when the triangle certificate yields a convex problem and shows that the exact nesting constraint can be nonconvex.

\begin{lemma}[Continuity of retained envelopes]\label{lem:envelope-continuity}
Let $d_H^{W_1}$ denote the Hausdorff distance induced by $W_1$.
For $B,B'\in\mathcal H$ and $\delta,\delta'\ge0$,
\begin{align}
 d_H^{W_1}(\mathcal A_\delta(B),\mathcal A_{\delta'}(B'))
 &\le W_1(B,B')+|\delta-\delta'|,\label{eq:envelope-continuity}\\
 |\Phi(B,\delta;Q_0)-\Phi(B',\delta';Q_0)|
 &\le W_1(B,B')+|\delta-\delta'|.\label{eq:phi-continuity}
\end{align}
\end{lemma}
\Needspace{4\baselineskip}
\begin{proof}
Given $Q\in\mathcal A_\delta(B)$, put $a=W_1(Q,B')$.
If $a\le\delta'$, keep $Q$.
Otherwise mix it toward $B'$: $Q'=(\delta'/a)Q+(1-\delta'/a)B'\in\mathcal H$.
Convexity of transport cost gives $W_1(Q',B')\le\delta'$ and $W_1(Q,Q')\le a-\delta'\le\delta+W_1(B,B')-\delta'$.
Interchange the pairs to obtain the Hausdorff bound.
The objective $Q\mapsto W_1(Q,Q_0)$ is 1-Lipschitz, proving the second bound.
\end{proof}

Let $\mathcal B\subseteq\mathcal H$ be compact and $\overline r:\mathcal B\to[0,\infty)$ continuous with $\overline r(B)\ge r^\star(B)$.
If a baseline is feasible, \eqref{rec:center-rule} attains its optimum: the lemma makes its feasible set closed.
For a parameterization $B=Q_\alpha$, compact indices alone do not suffice; require $\alpha\mapsto Q_\alpha$ to be $W_1$-continuous and the radius rule continuous.
For the same benchmarks and radius rule, every triangle-certified candidate is feasible for the exact criterion.
Consequently, the global optimum under the exact criterion cannot have a larger radius.

For the inexpensive rule \eqref{rec:certificate-rule}, compact convex indices and continuous convex certificates $\overline r(\alpha)$ and $\overline s(\alpha)$ give an attained convex program whenever a feasible baseline exists.
Such certificates can use \eqref{cal:covering-bound} and a same-observation coupling between imputations to bound movement.
This convexity is a property of the certificates on the specified index set; an arbitrary parameterization need not preserve it.

\paragraph{Exact feasibility need not be convex.}
On $\X^\star=\{-1,0,2\}$ with ordinary distance, let $\mathcal H=\mathcal P(\X^\star)$, $\mathcal F=\{\delta_{-1}\}$, $Q_0=\delta_0$, and $\delta_0^{\rm old}=4/3$.
Use the constant candidate radius 1 on $\operatorname{conv}\{\delta_{-1},\delta_0\}$; it covers $\mathcal F$ throughout.
For $B_1=\delta_{-1}$, $B_2=\delta_0$, and $B_m=(B_1+B_2)/2$, direct transport optimization gives
\[
 \Phi(B_1,1;Q_0)=4/3,\qquad
 \Phi(B_2,1;Q_0)=1,\qquad \Phi(B_m,1;Q_0)=3/2.
\]
For $B_1$ the budget permits mass $1/3$ at 2, maximizing $\E|X^\star|$ at $4/3$; for $B_2$ the objective is the budget itself.
For $B_m$, move the mass $1/2$ at 0 to 2, at cost 1, giving $3/2$; no higher value is possible because each unit of cost increases $|X^\star|$ by at most one from its benchmark expectation $1/2$.
Hence both endpoints are accepted while their midpoint is rejected.

\subsection{A strict gain beyond the triangle certificate}
\label{app:three-point}
This example shows how exact nesting can accept a smaller covering radius and tighter parameter bounds when the triangle certificate permits no benchmark change.
It also shows that the corresponding Lipschitz tubes need not be nested.

Take $u=(-1,0)$, $v=(1,0)$, $w=(0,1)$ with Euclidean distance, $\mathcal H=\mathcal P(\{u,v,w\})$, and $\mathcal F=\operatorname{conv}\{\delta_u,\delta_v\}$.
The initial benchmark $Q_0=\delta_w$ has exact covering radius $\delta_0=\sqrt2$.
For a candidate $B=(p_u,p_v,p_w)$,
\[
 r^\star(B)=2\max(p_u,p_v)+\sqrt2 p_w,\qquad
 W_1(B,Q_0)=\sqrt2(1-p_w).
\]
The triangle certificate requires their sum, $\sqrt2+2\max(p_u,p_v)$, to be at most $\delta_0=\sqrt2$.
This forces $p_u=p_v=0$, so only $Q_0$ is accepted.

The midpoint $B=(1/2,1/2,0)$ has $r^\star(B)=1$, attaining the lower bound $W_1(\delta_u,\delta_v)/2$ from the triangle inequality.
Exact nesting accepts it: $\Phi(B,1;Q_0)=\sqrt2=\delta_0$, since all laws have distance at most $\sqrt2$ from $Q_0$, and $\delta_u$ belongs to the new envelope and attains this distance.

For $g_\theta(x)=\theta-d(x,w)$ and $\Theta=\mathbb R$, $\E_Qd(X^\star,w)=\sqrt2(1-Q\{w\})$.
The sharp set is therefore $\{\sqrt2\}$.
The old envelope contains every law on the three points, giving $[0,\sqrt2]$.
In the new envelope, mass at $w$ is at most $1/\sqrt2$: each unit moved there from $u$ or $v$ costs $\sqrt2$.
Moving that amount from $B$ attains the bound; $B$ itself has zero mass at $w$.
Thus
\begin{equation}
 \Theta_{\rm DRO}(\sqrt2;Q_0)=[0,\sqrt2],\qquad
 \Theta_{\rm DRO}(1;B)=[\sqrt2-1,\sqrt2].
 \label{eq:three-point-gain}
\end{equation}
Both the radius and the parameter set are strictly smaller.
The moment calculation establishes strict parameter tightening; a smaller radius alone does not guarantee it.

The individual 1-Lipschitz tubes are $[-\sqrt2,\sqrt2]$ and $[\sqrt2-1,\sqrt2+1]$, respectively.
They are not nested, but intersecting them gives the valid nested refinement $[\sqrt2-1,\sqrt2]$.

\subsection{One finite program for enclosing centers}
\label{app:recenter-finite}\label{app:radius-finite}
\label{op:finite-center}\label{app:finite-center-exact}
This subsection provides a finite linear-program implementation of the two benchmark-selection problems in Step~2 of Algorithm~\ref{alg:main-dro-construction}.

On finite $\X^\star=\{x_1^\star,\ldots,x_N^\star\}$, let $D_{ij}=d(x_i^\star,x_j^\star)$, let $\mathsf H$ be the probability vectors satisfying the retained linear restrictions, and suppose a complete generating list is known: $\mathcal F=\operatorname{conv}\{q^{(1)},\ldots,q^{(K)}\}$.
Mixtures of generator couplings show $r^\star(b)=\max_k W_1(q^{(k)},b)$; the reverse inequality holds because every generator is compatible.
Thus
\begin{equation}
 \begin{aligned}
 \min_{b,r,\Pi^{(1)},\ldots,\Pi^{(K)}}\quad&r\\
 \text{subject to}\quad &b\in\mathsf H,\quad r\ge0,\quad\Pi^{(k)}\ge0,\\
 &\Pi^{(k)}\mathbf1=q^{(k)},\quad (\Pi^{(k)})^\top\mathbf1=b,
 \quad\langle D,\Pi^{(k)}\rangle\le r\quad(\text{all }k)
 \end{aligned}\label{app:finite-center-lp}
\end{equation}
is the exact attained enclosing-center LP.
Here $\mathbf1$ is the all-ones vector and $\langle D,\Pi\rangle=\sum_{i,j}D_{ij}\Pi_{ij}$.
To impose triangle-certified nesting at an old mass vector $b^0$, add
\begin{equation}
 \Gamma\ge0,\quad\Gamma\mathbf1=b,\quad\Gamma^\top\mathbf1=b^0,
 \qquad r+\langle D,\Gamma\rangle\le\delta_0.
 \label{app:finite-center-nest-lp}
\end{equation}
Feasible plans certify both coverage and movement.
Conversely, optimal transport plans realize those distances, proving exactness for the triangle rule.
Compactness and the feasible baseline give attainment.
Enumerating all generators may be costly: a sampled subset does not certify coverage of the full class without an additional approximation bound.
This construction supplies calibration, unrestricted center proposals, and the convex incumbent in Algorithm~\ref{alg:main-dro-construction}.

\paragraph{Intersection fallback.}
\label{app:intersection-refinement}
When nesting is inconclusive, intersecting calibrated envelopes preserves $\mathcal F$.
The intersection is compact and convex, so Theorem~\ref{main:feasibility} applies.
On finite support, use one common mass vector and one coupling per benchmark.
Requiring a common law gives bounds at least as tight as intersecting the separate parameter sets; tube intersections need not give the same bounds.

%% file: appendix/computing_parameter_sets.tex
\clearpage
\section{Computing parameter sets}
\label{app:computational-analysis}
Step~4 of Algorithm~\ref{alg:main-dro-construction} combines moment bounds at individual parameter values with a search over parameters.
Subsection~\ref{app:finite-program} gives finite programs for the bounds; Subsection~\ref{app:parameter-computation} uses these or other moment routines for the search.
Subsection~\ref{app:dual-computation} bounds discretization error for the ambient envelope; Subsection~\ref{app:finite-nesting} supplies nesting certificates for Step~3.

\input{appendix/finite_support_parameter_feasibility}

\subsection{Affine endpoints and certified set inversion}
\label{app:affine-computation}\label{app:parameter-computation}
Here the search accounts for numerical uncertainty, whether locating the endpoints of an interval or approximating a possibly disconnected parameter set.

For $g_\theta=a-\theta b$, assume $\kappa\le\inf_{Q\in\mathcal A}\E_Qb$ for some $\kappa>0$, with continuous $a,b$ and compact convex envelope $\mathcal A$.
The positive denominator bound must hold over the whole envelope.
Continuity and mixtures give
\begin{equation}
 \Theta_{\rm DRO}=\Theta\cap[\theta_-,\theta_+],\qquad
 \theta_- =\min_{Q\in\mathcal A}\frac{\E_Qa}{\E_Qb},\quad
 \theta_+ =\max_{Q\in\mathcal A}\frac{\E_Qa}{\E_Qb}.
 \label{app:affine-computation-endpoints}
\end{equation}
Writing $m(\theta)=\min_Q\E_Qg_\theta$ and $M(\theta)=\max_Q\E_Qg_\theta$, the endpoints are their respective unique zeros.
For $\theta'>\theta$ and $f=m,M$,
\[
 \kappa(\theta'-\theta)\le f(\theta)-f(\theta')
 \le\bigl(\max_{Q\in\mathcal A}\E_Qb\bigr)(\theta'-\theta).
\]
This follows by comparing each law's affine moment and taking extrema.
Certified enclosures $\ell\le f(c)\le u$ permit bisection: positive $\ell$ raises the bracket's lower endpoint, negative $u$ lowers its upper endpoint, and otherwise the root lies in $[c+\ell/\kappa,c+u/\kappa]$.
An oracle width at most $\kappa\eta$ therefore yields endpoint brackets of width $\eta$.

For general scalar moments, interval set inversion is a standard alternative~\cite{JaulinWalter1993}.
Fix a nondegenerate compact search interval $J\subseteq\Theta$, a certified moment routine, and a known $K$ with $\sup_{Q\in\mathcal A}|\E_Q(g_\theta-g_{\theta'})| \le K|\theta-\theta'|$ on $J$.
Then $h(\theta)=\max\{m(\theta),-M(\theta)\}$ is $K$-Lipschitz and $\mathcal S_J=\Theta_{\rm DRO}\cap J=\{h\le0\}$.
Moment enclosures of width at most $\omega$ give $\ell(c)\le h(c)\le u(c)$ with $u(c)-\ell(c)\le\omega$.
For an interval with midpoint $c$ and half-width $r$, discard it if $\ell(c)-Kr>0$, accept it if $u(c)+Kr\le0$, and otherwise bisect until its width is at most $\eta$, when it remains unresolved.

\begin{samepage}
\begin{proposition}[Set-inversion guarantee]
\leavevmode\label{app:parameter-computation-guarantee}
If each certified oracle call terminates, subdivision terminates, and the unions of accepted intervals $\mathcal I$ and of accepted and unresolved intervals $\mathcal O$ satisfy
\begin{equation}
 \mathcal I\subseteq\mathcal S_J\subseteq\mathcal O,
 \qquad [h(\theta)]_+\le K\eta+\omega\quad(\theta\in\mathcal O).
 \label{app:eq:parameter-enclosures}
\end{equation}
As $\eta,\omega\to0$, $\mathcal O$ converges in Hausdorff distance to nonempty $\mathcal S_J$; for empty $\mathcal S_J$ it is eventually empty.
\end{proposition}
\end{samepage}
\begin{proof}
Every point of the interval satisfies $\ell(c)-Kr\le h(\theta)\le u(c)+Kr$, proving the tests.
There are at most $2^{1+\max\{0,\lceil\log_2(|J|/\eta)\rceil\}}-1$ calls.
On an unresolved interval $\ell(c)\le Kr$ and $2r\le\eta$, so $h(\theta)\le2Kr+\omega\le K\eta+\omega$.
If outer points stayed a positive distance from $\mathcal S_J$, compactness would supply a convergent subsequence whose limit has $h\le0$, a contradiction.
If the set is empty, $\min_Jh>0$ gives eventual emptiness.
\end{proof}
Resolution $\eta$ alone guarantees neither a distance error bound nor a component count.
An isolated feasible value remains covered even if the inner approximation is empty.
The certified oracle and $K$ are additional computational inputs.

\subsection{Discretizing the ambient transport dual}
\label{app:dual-computation}
A finite grid need not preserve valid moment bounds.
The error bounds below account for approximating both the latent space and the benchmark, when no additional restrictions are retained.

For $\mathcal H=\mathcal P(\X^\star)$ and an $L$-Lipschitz $g$, denote the upper moment by $U(g)$.
In \eqref{main:dual} it suffices to search $0\le\lambda\le L$: for $\lambda\ge L$ the inner maximum equals $g(z)$, and increasing $\lambda$ cannot improve the objective.
Choose a finite $\rho$-net $\{x_j\}_{j=1}^N$ in $\X^\star$ and a finite benchmark $\widehat Q_0=\sum_i b_i\delta_{z_i}$ with $W_1(\widehat Q_0,Q_0)\le\zeta$.
Define the LP
\begin{equation}
 v_N=\min_{0\le\lambda\le L,t}
 \left\{\lambda\delta+\sum_i b_it_i:
 t_i\ge g(x_j)-\lambda d(x_j,z_i)\ (\text{all }i,j)\right\}.
 \label{app:eq:finite-dual-lp}
\end{equation}
Then
\begin{equation}
 v_N-L\zeta\le U(g)\le v_N+2L\rho+L\zeta.
 \label{app:eq:finite-dual-certificate}
\end{equation}
To prove this, the net changes an inner maximum by a one-sided error between zero and $(L+\lambda)\rho\le2L\rho$; that maximum is $\lambda$-Lipschitz in $z$, so the benchmark change contributes at most $L\zeta$.
Take minima over $[0,L]$.
For certified solver bounds $v^-\le v_N\le v^+$ substitute $v^-$ below and $v^+$ above.
The enclosure width is at most $(v^+-v^-)+2L\rho+2L\zeta$.
Apply the same construction to $-g$ for the lower moment.
The bounds also hold when $\delta=0$ or $L=0$.
These finite convex formulations build on Wasserstein duality and reformulations~\cite{BlanchetMurthy2019,GaoKleywegt2023,MohajerinEsfahaniKuhn2018}.
The stated net and benchmark errors must be certified; compactness alone gives no useful computational complexity bound.
Retained restrictions require a constrained program, such as Appendix~\ref{app:regression-dual}.

\subsection{Computing and bounding the nesting test}
\label{app:finite-nesting}
The preceding subsections compute parameter bounds for a given envelope.
This subsection checks whether a proposed envelope lies inside the original one, as required in Step~3 of Algorithm~\ref{alg:main-dro-construction}.
The methods compute $\Phi$ or a certified upper bound $U_\Phi$; $U_\Phi\le\delta_0$ certifies nesting, with coverage certified in Step~2.

If $Q_0=\delta_{z_0}$, transport to the old benchmark is unique:
\begin{equation}
 \Phi(B,\delta;Q_0)=
 \max_{Q\in\mathcal A_\delta(B)}\E_Qd(X^\star,z_0).
 \label{eq:pointmass-phi}
\end{equation}
Thus exact acceptance is an ordinary constrained moment calculation.
Appendix~\ref{app:network-nesting-lp} supplies a related exact continuous-graph formulation for deterministic network imputations.

For a common finite state space $\X^\star=\{x_1^\star,\ldots,x_N^\star\}$, write $b,q_0$ for the masses of $B,Q_0$, respectively, and $D_{ij}=d(x_i^\star,x_j^\star)$.
Let $\mathsf H$ denote the probability vectors satisfying the retained restrictions.
Kantorovich--Rubinstein duality gives
\begin{equation}
 \Phi(B,\delta;Q_0)=\max_{q,\pi,\varphi}
     \sum_i\varphi_i(q_i-q_{0i}),
 \label{eq:finite-phi-bilinear}
\end{equation}
subject to $q\in\mathsf H$, $\pi\ge0$, $\pi\mathbf1=q$, $\pi^\top\mathbf1=b$, $\langle D,\pi\rangle\le\delta$, and $\varphi_a=0$, $|\varphi_i-\varphi_j|\le D_{ij}$ for a fixed anchor $a$.
The expression is bilinear, not generally an LP.
The normalized potential polytope is compact.
Enumerating its vertices and solving one moment LP at each vertex computes the exact maximum for small support, since for each fixed $q$ some potential vertex is optimal.

\Needspace{5\baselineskip}
For an upper bound, put $L_i=-D_{ia}$ and $U_i=D_{ia}$, replace $q_i\varphi_i$ by $z_i$, and impose the four McCormick inequalities \cite{McCormick1976}:
\begin{equation}
 \begin{aligned}
 z_i&\ge L_iq_i,&z_i&\ge\varphi_i+U_iq_i-U_i,\\
 z_i&\le U_iq_i,&z_i&\le\varphi_i+L_iq_i-L_i.
 \end{aligned}\label{eq:mccormick-phi}
\end{equation}
With linear retained restrictions, maximizing $\sum_i z_i-q_0^\top\varphi$ is an LP whose value $U_{\rm MC}$ is at least $\Phi$: every exact product is feasible because $q_i\in[0,1]$ and $\varphi_i\in[L_i,U_i]$.
Take the minimum with the triangle upper bound.
The objective at a dual-feasible LP solution, adjusted for certified numerical error, can establish acceptance; a feasible law and potential in \eqref{eq:finite-phi-bilinear} give a lower bound and cannot do so.
Equality at the acceptance boundary requires an exact or outward-certified upper bound, not a solver tolerance.
A grid of a continuous space alone gives no such upper certificate.

\begin{samepage}
\paragraph{A benchmark with two support points accepted by the LP relaxation.}
This example shows that the LP relaxation can certify nesting when the triangle test is inconclusive, allowing a smaller covering radius and tighter parameter bounds.
Both benchmarks here have two support points; Appendix~\ref{app:three-point} uses a point-mass initial benchmark.

Let $\X^\star=\{0,-1,2\}$ with ordinary distance, $\mathcal H=\mathcal P(\X^\star)$, and $\mathcal F=\operatorname{conv}\{\delta_0,\delta_{-1}\}$.
With masses ordered as $(0,-1,2)$, take $q_0=(0,2/3,1/3)$ with radius $4/3$, and $b=(1/6,5/6,0)$ with radius $5/6$.
Both radii are exact; $W_1(b,q_0)=5/6$, so the triangle sum is $5/3>4/3$.
For $q=(x,y,z)$, the new envelope is $2z\le y$.
Anchoring the potential at 0, the McCormick objective is at most
\[
 \left(\frac23-\frac y3\right)
 +\left(\frac23+\frac{2z}3\right)\le\frac43.
\]
The two component bounds follow by convexly combining their McCormick upper inequalities.
At $q=\delta_0$, potentials $\varphi(-1)=-1$, $\varphi(2)=-2$ attain $4/3$.
Thus $U_{\rm MC}=\Phi=4/3$ certifies a pair the triangle rule rejects.
For the moment $\theta-\mathbf1\{X^\star=2\}$, the old and new DRO sets are $[0,7/9]$ and $[0,1/3]$; the sharp set is $\{0\}$.
\par
\end{samepage}

%% file: appendix/finite_support_parameter_feasibility.tex
\subsection{Parameter feasibility on finite supports}
\label{app:finite-program}
The formulation below tests membership in $\Theta_{\mathrm{DRO}}$ by expressing the transport budget, retained restrictions, and moment condition as constraints on a coupling.

Suppose the candidate law has maintained support $x_1^\star,\ldots,x_N^\star$, and the benchmark is $Q_0=\sum_{i=1}^M b_i\delta_{z_i^\star}$ with $b_i\geq0$ and $\sum_i b_i=1$.
Put $c_{ij}=d(x_j^\star,z_i^\star)$ and $g_{\theta j}=g_\theta(x_j^\star)$.
Write $\mathsf H$ for the probability vectors $q$ such that $\sum_jq_j\delta_{x_j^\star}\in\mathcal H$.
The feasibility constraints are
\begin{equation}
 \begin{aligned}
 &\pi_{ij}\geq0,\qquad
   \sum_j\pi_{ij}=b_i\quad(i=1,\ldots,M),\\
 &q_j=\sum_i\pi_{ij}\quad(j=1,\ldots,N),\qquad q\in\mathsf H,\\
 &\sum_{i,j}c_{ij}\pi_{ij}\leq\delta,\qquad
   \sum_jg_{\theta j}q_j=0.
 \end{aligned}
 \label{app:eq:finite-program}
\end{equation}
To compute a moment endpoint, remove the zero-moment equality and minimize or maximize $\sum_jg_{\theta j}q_j$.
To compute the profile, retain that equality, remove the budget inequality, and minimize $\sum_{i,j}c_{ij}\pi_{ij}$.
When the restrictions on $q$ are linear equalities and inequalities, these are linear programs for fixed $\theta$.

On these finite supports, the following equivalence holds for every $\theta\in\Theta$ and finite $\delta\ge0$:
\[
 \text{\eqref{app:eq:finite-program} is feasible}
 \quad\Longleftrightarrow\quad
 \theta\in\Theta_{\mathrm{DRO}}(\delta;Q_0).
\]

\begin{proof}[Proof of the equivalence]
Feasibility gives $q\ge0$ and $\sum_jq_j=1$.
For any $Q=\sum_jq_j\delta_{x_j^\star}$,
\[
 W_1(Q,Q_0)=\min_{\pi\ge0}
 \left\{\sum_{i,j}c_{ij}\pi_{ij}:
 \pi\mathbf1=b,\quad\pi^\top\mathbf1=q\right\}.
\]
Compactness of the nonempty coupling polytope gives attainment.
Together with $\E_Qg_\theta=\sum_jg_{\theta j}q_j$, this proves the equivalence and the endpoint and profile formulations.
\end{proof}

\begin{samepage}
Benchmark and candidate supports may differ.
The finite-program equivalence permits $Q_0\notin\mathcal H$ as an extension of the main setup, where Assumption~\ref{ass:retained-class} and Algorithm~\ref{alg:main-dro-construction} require $Q_0\in\mathcal H$.
Feasibility can fail at some radii.
Continuous-model approximations require a separate error bound (Appendix~\ref{app:dual-computation}).
\par
\end{samepage}

%% file: appendix/missing_outcomes.tex
\clearpage
\section{Missing outcomes}
\label{app:missing-outcomes}
\label{app:missingness}
\label{app:missing-dro-calculation}

This example separates the effects of the benchmark and the retained restrictions.
Let $X=(D,DY)$, with binary $D,Y$, and write $p_{ab}=P_X\{(a,b)\}$, so $p_{00}+p_{10}+p_{11}=1$.
The moment is $g_\theta(d,y)=\theta-y$, with $\Theta=[0,1]$.
We first use Hamming distance and the deliberately broad class $\mathcal H=\mathcal P(\{0,1\}^2)$.

\subsection{Compatibility, coverage, and moment optimization}
Compatibility fixes the two observed worker masses and the total nonworker mass.
The only free choice is the success probability among missing outcomes:
\begin{equation}
 \mathcal F=\{Q_t:0\le t\le1\},\qquad
 Q_t=p_{00}\{(1-t)\delta_{(0,0)}+t\delta_{(0,1)}\}
       +p_{10}\delta_{(1,0)}+p_{11}\delta_{(1,1)}.
 \label{app:eq:missing-completions}
\end{equation}
Every displayed law maps to $P_X$, and every compatible law has this form.
Thus $\E_{Q_t}Y=p_{11}+p_{00}t$, proving the familiar sharp interval $[p_{11},p_{11}+p_{00}]$ and attainment of every point \citep{Manski1989}.
If $p_{00}=0$, the laws coincide and no conditional success probability for nonworkers need be defined.

Couple $Q_s$ and $Q_t$ by keeping all common mass fixed and moving $p_{00}|s-t|$ between $(0,0)$ and $(0,1)$.
Its cost equals the absolute difference of the two outcome means, a lower bound for any coupling.
Consequently
\begin{equation}
 W_1(Q_s,Q_t)=p_{00}|s-t|,\qquad
 r^\star(Q_t)=p_{00}\max\{t,1-t\}.
 \label{app:eq:missing-center-radius}
\end{equation}
The initial zero imputation $Q_0$ therefore has exact covering radius $\delta_0=p_{00}$.

For any $\delta\ge0$, the outcome mean over $\mathcal A_\delta(Q_t)$ lies within $\delta$ of $p_{11}+p_{00}t$ and in $[0,1]$.
Both endpoint bounds are attainable: flipping binary outcome mass at fixed $D$ changes its mean and transport cost by the same amount.
Writing $\mu_t=p_{11}+p_{00}t$ gives
\[
 \min_{Q\in\mathcal A_\delta(Q_t)}\E_Qg_\theta
 =\theta-\min\{1,\mu_t+\delta\},\qquad
 \max_{Q\in\mathcal A_\delta(Q_t)}\E_Qg_\theta
 =\theta-\max\{0,\mu_t-\delta\}.
\]
Theorem~\ref{main:feasibility} and $L(\theta)=1$ therefore give
\begin{equation}
 \Theta_{\mathrm{DRO}}(\delta;Q_t)
 =\mathcal T_\delta(Q_t)
 =[\max\{0,\mu_t-\delta\},\min\{1,\mu_t+\delta\}].
 \label{app:eq:missing-recentered-set}
\end{equation}

\subsection{Midpoint geometry and the retained-information comparator}
\label{app:recentering-missing}
The midpoint $Q_{1/2}$ randomizes a binary missing outcome; it does not replace that outcome by the number $1/2$.
It has radius $p_{00}/2$ and movement $W_1(Q_{1/2},Q_0)=p_{00}/2$, so the triangle certificate is exact.
Moreover, for every benchmark $B\in\mathcal H$,
\[
 p_{00}=W_1(Q_0,Q_1)
 \le W_1(Q_0,B)+W_1(B,Q_1)\le2r^\star(B).
\]
Hence the midpoint minimizes the covering radius globally:
\begin{equation}
 \inf_{B\in\mathcal H}r^\star(B)=r^\star(Q_{1/2})=p_{00}/2.
 \label{app:eq:missing-optimal-radius}
\end{equation}
At this calibrated pair, \eqref{app:eq:missing-recentered-set} becomes the sharp interval.
For $(p_{00},p_{10},p_{11})=(0.2,0.4,0.4)$, the initial interval and tube are $[0.2,0.6]$; both become $[0.4,0.6]$ after recentering.

An equally inexpensive improvement is to retain the full observable marginal from the outset:
\[
 \mathcal H_{\mathrm{obs}}
 =\{Q:Q_{(D,DY)}=P_X\}=\mathcal F.
\]
All these laws already lie in the initial radius-$p_{00}$ ball.
The initial constrained moment programs therefore give the sharp interval, and recentering leaves it unchanged.
Retaining only the $D$ marginal would not fix the two worker outcome masses and does not give this conclusion.
The comparison shows why benchmark refinement and the choice of retained restrictions should be assessed together: the broad initial class makes the effect of recentering visible, while the observable constraints remove the same excess width directly.

%% file: appendix/thresholded_networks.tex
\clearpage
\section{Thresholded networks}
\label{app:network}
\label{app:network-calculation}

This example implements capacity-based coverage, moment optimization, and certified benchmark refinement.
Analytical endpoints validate the programs without being required to solve them.
Capacity comparators show how additional retained information improves the same procedure when sharp identification is difficult.

\subsection{Measurement model and the source of sharp-computation difficulty}
There are $m\ge2$ treated core individuals, a treated contact $h$, and an untreated contact $s$, with $n=m+2$.
Orient each edge of a known weighted core graph $E$ from its smaller to its larger label.
Write $\operatorname{out}(i)$ for the edges leaving $i$ and $d_i=\sum_{e\in\operatorname{out}(i)}w_e$, $d_{\max}=\max_i d_i>0$, and $v=\sum_{e\in E}w_e$.
For $x\in[0,1]^E$, define the row-stochastic graph
\begin{equation}
 \begin{aligned}
 A_{ij}(x)&=\eta w_ex_e &&(e=\{i,j\}\in E,\ i<j),\\
 A_{is}(x)&=1-\eta\sum_{e\in\operatorname{out}(i)}w_ex_e
                                      &&(i=1,\ldots,m),\\
 A_{h1}(x)&=A_{s1}(x)=1,
 \end{aligned}
 \label{app:eq:network-construction}
\end{equation}
with all other entries zero.
Assume $0<r:=\eta d_{\max}<\min\{\xi,1-\xi\}$.
Each eligible weight lies below the reporting threshold $\xi$; each core weight to $s$ exceeds it.
Thus every $A(x)$ gives the same binary report $G=A(0)$.

The maintained hidden-role restriction is $x_e=x_e(z)=\mathbf1\{z_i\ne z_j\}$ for a binary vector $z$.
Conditional on $X=(Y,T,G)$, any distribution on
\begin{equation}
 \Gamma(G)=\{A(x(z)):z\in\{0,1\}^m\}
 \label{app:eq:network-fiber}
\end{equation}
is allowed.
The known capacities and shared roles are additional restrictions beyond thresholding.
The auxiliary labels may be swapped across observed records; all constructions use their record-specific labels.
Put $a=\E_{P_X}[T^\top(Y-\beta_0T)]>0$ and $b(T,A)=T^\top AT$.
The maintained scalar moment is $\E_Qg_\theta=a-\theta\E_Qb=0$.

Let $C(z)=\sum_e w_ex_e(z)$ and $C^\star=\max_z C(z)$.
The treated contact contributes one to exposure, the untreated contact zero, and the core rows $\eta C(z)$.
Consequently
\begin{equation}
 b(T,A(x(z)))=1+\eta C(z),\qquad
 \{\E_Qb:Q\in\mathcal F\}=[1,1+\eta C^\star].
 \label{app:eq:network-exposure}
\end{equation}
Constant roles and a maximizing cut attain the endpoints at every observation.
Mixtures preserve $P_X$ and attain the intermediate values.
With $\Theta=\mathbb R$ as in Example~\ref{ex:lim-threshold}, $a>0$ forces $\theta>0$, so $\Theta_{\mathrm{ID}}=[a/(1+\eta C^\star),a]$.
For integer weights, the choice $a=1$, $\xi=2/5$, and $\eta=[4(m+1)d_{\max}]^{-1}$ has polynomial encoding length.
The sharp lower endpoint recovers $C^\star$ exactly, so computing it is NP-hard \citep{GareyJohnsonStockmeyer1976}.
The difficulty is the common role assignment across rows.
This result concerns the scalar moment.
The motivating conditional outcome restriction $\E[Y-\beta_0T-\theta G^\star T\mid G^\star,T]=0$ is stronger and would define a different compatible class.

\paragraph{Numerical instance.}
We generate the core graph on $m=150$ nodes by independently including each unordered pair with probability $0.1$ and assigning independent weights uniformly distributed on $\{1,\ldots,10\}$ to the included edges.
Figure~\ref{fig:recentering-sets} conditions on one fixed realization, with $n=152$, $1{,}112$ eligible edges, $v=6{,}143$, and $d_{\max}=115$, with $\xi=0.4$, $\beta_0=0$, $r=1/604$, and $\eta=1/69{,}460$.
A fair coin swaps the auxiliary labels; conditional on $(T,G)$, the outcomes are independent with
\begin{equation}
 Y_i\sim\operatorname{Bernoulli}
       \left(\frac{T_i}{m+1}+\frac{1-T_i}{5}\right).
 \label{app:eq:network-observed}
\end{equation}
This observed law gives $a=1$.
Both auxiliary nodes have exposure one, but their outcome means are $1/151$ and $1/5$; hence no common $\theta$ satisfies the stronger conditional outcome model.
The illustration maintains only the scalar moment.

\subsection{Coverage and constrained moment programs}
\label{app:radius-network}
Use the capacity weights and metric
\[
 \omega_i=\frac1{10}+\frac9{10}\frac{d_i}{d_{\max}}
 \ (i\le m),\qquad \omega_h=\omega_s=1,\qquad
 d_{\mathrm{cap}}(A,B)=\max_{i,j}\frac{|A_{ij}-B_{ij}|}{\omega_i}.
\]
Let $\mathcal S=\{A\ge0:A_{ii}=0,\ A\mathbf1=\mathbf1\}$ and retain $\mathcal H=\{Q\in\mathcal P(\operatorname{supp}(P_X) \times\mathcal S):Q_X=P_X\}$, with bounded outcomes and distance
\[
 d((x,A),(x',B))=\|y-y'\|_\infty
  +10\mathbf1\{(t,g)\ne(t',g')\}+d_{\mathrm{cap}}(A,B).
\]
The class is compact and convex.
It retains the observed marginal and row normalization, while relaxing the role and threshold restrictions.
Its graph diameter is at most $10$.

Let $C(T,G)$ be a deterministic imputation independent of outcomes, and let $B_C=\mathcal L_{P_X}(X,C(T,G))$ be its joint law under $P_X$.
The same-record coupling gives
\begin{equation}
 W_1(Q,B_C)=\E_Qd_{\mathrm{cap}}(G^\star,C(T,G))
 \qquad(Q\in\mathcal H).
 \label{app:eq:network-conditional-distance}
\end{equation}
To prove the reverse bound for any coupling, the triangle inequality bounds $d_{\mathrm{cap}}(A,C(T,G))$ by $d_{\mathrm{cap}}(A,C(T',G'))+10\mathbf1\{(T,G)\ne(T',G')\}$.
Integrate and use the definition of the full distance.
This identity does not extend automatically to arbitrary randomized or outcome-dependent imputations.

Initially, $C=G$ and $Q_0=B_G$.
Every compatible graph satisfies
\[
 d_{\mathrm{cap}}(A(x(z)),G)
 =\eta\max_i\frac{\sum_{e\in\operatorname{out}(i)}w_ex_e(z)}{\omega_i}
 \le\eta\max_i\frac{d_i}{\omega_i}=r.
\]
Separate a largest-capacity row from all its outgoing neighbors to attain equality.
Hence $r^\star(Q_0)=r$: coverage can be computed from row capacities without solving maximum cut.

For the imputation family $C_t=A(t\mathbf1)$, $0\le t\le1/2$, set $B_t=B_{C_t}$ and $\delta_t=(1-t)r$.
All observed categories differ only by the auxiliary label swap; $C_t$ respects this swap.
Undoing it reduces moment optimization to a law on $\mathcal S$ with an expected transport budget.
Since exposure is linear and $d_{\mathrm{cap}}(\cdot,C_t)$ is convex, replacing that graph law by its mean preserves exposure and weakly lowers cost.
Thus the two exposure extrema are exactly the following finite linear programs, with continuous graph entries:
\begin{equation}
 \underset{A}{\operatorname{min/max}}\ T^\top AT
 \quad\text{subject to}\quad
 A\ge0,\quad A_{ii}=0,\quad A\mathbf1=\mathbf1,\quad
 |A_{ij}-(C_t)_{ij}|\le\omega_i\delta_t.
 \label{app:eq:network-moment-program}
\end{equation}
Here $T$ uses the common relabelled template.
This reduction uses linearity of exposure; the nesting calculation below has a convex objective and needs a law-level program.
For any proposed $\theta>0$, the optimized exposures $\underline b_t,\overline b_t$ give moment bounds $a-\theta\overline b_t$ and $a-\theta\underline b_t$.
The parameter search therefore uses the same moment-optimization interface as the other applications.

The following endpoint witnesses validate these programs.
Writing $S_\omega=\sum_{i=1}^m\omega_i$ and $D=d_{\mathrm{cap}}(A,C_t)$, row normalization gives $1-D\le b(T,A)\le1+t\eta v+D S_\omega$.
For the lower endpoint, put all core mass on $s$ and move $\delta_t$ of the contact $h$'s mass from node $1$ to $s$.
For the upper endpoint, move $\omega_i\delta_t$ from $s$ to $h$ in each core row of $C_t$, keeping the auxiliary rows fixed.
Both graphs are feasible: $t\eta d_i\le\omega_i\delta_t$ and $t\eta d_i+\omega_i\delta_t\le r<1$.
Therefore
\begin{equation}
 \underline b_t=1-\delta_t,\qquad
 \overline b_t=1+t\eta v+\delta_t S_\omega,\qquad
 \Theta_{\mathrm{DRO}}(\delta_t;B_t)
 =\left[\frac a{1+t\eta v+\delta_t S_\omega},
         \frac a{1-\delta_t}\right].
 \label{app:eq:network-mean-lp}
\end{equation}
\previewequationalias{app:eq:network-sets}

\subsection{Benchmark selection, tubes, and metric sensitivity}
\label{app:recentering-network}
For the same family,
\begin{equation}
 W_1(B_t,B_u)=|t-u|r,\qquad
 r^\star(B_t)=(1-t)r,\qquad
 Q_{\mathrm c}=B_{1/2}.
 \label{app:eq:network-midpoint}
\end{equation}
The distance follows from \eqref{app:eq:network-conditional-distance}.
Every cut graph differs from $C_t$ by at most $(1-t)\eta d_i$ in row $i$; the largest-capacity-row cut attains the bound.
These exact coverage and movement certificates make $\delta_t+W_1(B_t,Q_0)=r$ throughout the family, so the inexpensive triangle rule selects $t=1/2$.
The zero-cut law and the largest-row cut law are compatible and at distance $r$.
Their triangle inequality through any benchmark $B$ gives $r\le2r^\star(B)$.
Hence
\begin{equation}
 \inf_{B\in\mathcal H}r^\star(B)=r^\star(Q_{\mathrm c})=r/2.
 \label{app:eq:network-optimal-radius}
\end{equation}
The sharper nesting criterion cannot further reduce this covering radius.
Its role here is to make benchmark acceptance computable for other proposed imputations.

At the midpoint the exact DRO interval is
\begin{equation}
 \Theta_{\mathrm{DRO}}(r/2;Q_{\mathrm c})
 =\left[\frac a{1+\eta v/2+(r/2)S_\omega},\frac a{1-r/2}\right]
 \subsetneq
 \left[\frac a{1+rS_\omega},\frac a{1-r}\right].
 \label{app:eq:network-recentered-dro}
\end{equation}
Both endpoints improve.
For the lower endpoint the strict gain comes specifically from the positive floor in the metric:
\[
 rS_\omega-\eta v=\frac\eta{10}(md_{\max}-v)>0.
\]
The last inequality follows from $d_m=0$ and $d_{\max}>0$.
The chosen metric therefore affects how much parameter tightening recentering achieves.

The observed contribution is fixed in expectation over $\mathcal H$.
Replace $g_\theta$ by $a-\theta b$, which has the same expected moment.
Exposure is $(1+S_\omega)$-Lipschitz: for fixed $T$, use $b=\sum_{i:T_i=1}(1-A_{is})$; across treatment categories, $b\in[0,m+1]$ and $10(1+S_\omega)\ge m+1$.
The corresponding exposure tube is
\begin{equation}
 \mathcal T^{\mathrm{exp}}_{\delta_t}(B_t)
 =\{\theta>0:|a-\theta(1+t\eta v)|
                \le\delta_t(1+S_\omega)\theta\}.
 \label{app:eq:network-recentered-tubes}
\end{equation}
This tube uses the corrected moment; its Lipschitz bound need not apply to the original outcome moment.
When $1+t\eta v>\delta_t(1+S_\omega)$, divide $a$ by the two positive quantities $1+t\eta v\pm\delta_t(1+S_\omega)$ to obtain its endpoints.
For the numerical instance, $S_\omega=72{,}537/1{,}150$.
The moment programs and midpoint selection yield the initial and recentered DRO intervals $[0.905445,1.001658]$ and $[0.912047,1.000829]$ shown in Figure~\ref{fig:recentering-sets}; \eqref{app:eq:network-mean-lp} validates their endpoints, and \eqref{app:eq:network-recentered-tubes} gives the accompanying tubes.

\subsection{Strengthening the retained restrictions}
\label{app:network-comparators}
\label{app:residual-network}
The broad class above intentionally displays the effect of benchmark refinement.
A simple comparator retains the known capacity support:
\[
 \mathcal H_{\mathrm{cap}}
 =\{Q:Q_X=P_X,\quad G^\star=A(x),\ x\in[0,1]^E\ \ Q\text{-a.s.}\}.
\]
This is a closed convex class of laws; it drops the shared-role restriction but keeps all edge capacities and fixed auxiliary links.
Every such graph lies within $r$ of $G$ and within $r/2$ of $C_{1/2}$.
Thus the transport constraint is redundant for both calibrated pairs.
The exposure range is $[1,1+\eta v]$, attained at $x=0$ and $x=\mathbf1$, so both yield
\begin{equation}
 \Theta_{\mathrm{cap}}
 =\left[\frac a{1+\eta v},a\right].
 \label{app:eq:network-capacity-bound}
\end{equation}
The same moment programs can impose these additional linear support constraints.
There is no further gain from recentering this class.

A cut upper certificate $C^\star\le U$ gives the further valid restriction $\E_Qb\le1+\eta U$ and outer interval $[a/(1+\eta U),a]$.
For edge-disjoint triangles, every cut omits at least one edge from each triangle.
Subtracting the sum of their minimum weights from $v$ supplies such a certificate using only integer arithmetic.
These comparisons do not claim that the broad DRO relaxation dominates problem-specific relaxations; they show how inexpensive information strengthens its inputs.
In the numerical instance, a feasible cut has weight $4{,}167$, while $187$ edge-disjoint triangles have minimum weights summing to $720$.
Thus $4{,}167\le C^\star\le5{,}423$.
The sharp lower endpoint lies in $[69460/74883,69460/73627]\simeq[0.927580,0.943404]$, the bracket shown in Figure~\ref{fig:recentering-sets}; the sharp upper endpoint is $1$.
An exact maximum cut is not needed for these comparisons.

\subsection{An exact law-level nesting program}
\label{app:network-nesting}
\label{app:network-nesting-lp}
For a deterministic benchmark $B_C$ as in \eqref{app:eq:network-conditional-distance}, the old-distance identity gives
\[
 \Phi(B_C,\delta;Q_0)
 =\max_{Q\in\mathcal A_\delta(B_C)}
      \E_Qd_{\mathrm{cap}}(G^\star,G).
\]
Thus exact acceptance is another constrained moment calculation.
For $B_t$ at its calibrated radius, the triangle upper bound and the largest-row cut give $\Phi(B_t,\delta_t;Q_0)=r$.
Assume additionally that $C$ respects the example's auxiliary label swaps.
After undoing those swaps, its matrix is a common $C\in\mathcal S$, the old graph is a common $G$, and the problem is over graph laws $\mu$ with $\E_\mu d_{\mathrm{cap}}(A,C)\le\delta$.
Independence of the imputation from outcomes is essential for this reduction; arbitrary randomized benchmarks require coupling variables beyond those below.

Each row of $G$ has a single unit entry, at column $j_i$.
Row normalization therefore gives
\[
 h(A):=d_{\mathrm{cap}}(A,G)
 =\max_{1\le i\le n}\ell_i(A),\qquad
 \ell_i(A)=\frac{1-A_{ij_i}}{\omega_i}.
\]
Introduce branch probabilities $\lambda_i$, scaled graphs $Z^i$, and transport costs $\tau_i$.
The exact finite LP is
\begin{equation}
 \begin{aligned}
 \Phi(B_C,\delta;Q_0)
 =\max_{\lambda,Z,\tau}\quad&
       \sum_{i=1}^n\frac{\lambda_i-Z^i_{ij_i}}{\omega_i}\\
 \text{subject to}\quad&
       \lambda_i,\tau_i\ge0,\quad Z^i\ge0,\quad
       \sum_i\lambda_i=1,\quad\sum_i\tau_i\le\delta,\\
 & Z^i_{kk}=0,\quad\sum_jZ^i_{kj}=\lambda_i
                   &&(i,k=1,\ldots,n),\\
 &-\omega_k\tau_i\le Z^i_{kj}-\lambda_iC_{kj}
             \le\omega_k\tau_i &&(i,k,j=1,\ldots,n).
 \end{aligned}
 \label{app:eq:network-nesting-lp}
\end{equation}
No graph-space discretization is used.
The shared budget $\sum_i\tau_i\le\delta$ controls expected transport cost.
Requiring $d_{\mathrm{cap}}(A,C)\le\delta$ for every graph can exclude feasible mixtures and underestimate $\Phi$.

To prove equality, partition any feasible graph law by an index attaining $\max_i\ell_i(A)$, breaking ties measurably.
On branch $i$, let $\lambda_i$ be its probability, $Z^i$ its unnormalized matrix mean, and $\tau_i$ its integrated distance to $C$.
Convexity of the distance gives the displayed constraints, and the objective is exactly $\E_\mu h(A)$.
Conversely, from any LP solution put mass $\lambda_i$ on $Z^i/\lambda_i$ for every positive branch.
Its expected cost is at most $\sum_i\tau_i$; its expected $h$ is at least the LP objective, because $h\ge\ell_i$ on each branch.
The two inequalities prove equality.
Linear restrictions on the mean graph can be imposed on $\sum_iZ^i$; linear support restrictions apply to each scaled branch.
The LP has order $n^3$ variables and is not asserted to be inexpensive for every network size.
The numerical midpoint needs only the triangle certificate.

%% file: appendix/measurement_error_regression.tex
\section{Measurement-error regression}
\label{app:measurement-error}
\label{app:regression-calculation}

This example applies Algorithm~\ref{alg:main-dro-construction} to a regression slope with a mismeasured regressor.
Retaining the observed joint distribution and the squared-error budget makes Step~4 a tractable sharp-identification problem.
The numerical illustration follows all five steps: the exact nesting test accepts a smaller-radius benchmark, while the DRO set is already sharp.

\subsection{Model and retained information}
\label{ex:additive-measurement-error}
Let $Y$ be an accurately measured outcome, $U$ the true regressor, and $W=U+\epsilon$ its recorded value.
In the notation of Section~\ref{sec:model},
\[
 X^\star=(Y,U,\epsilon),\qquad
 X=f(X^\star)=(Y,U+\epsilon)=(Y,W).
\]
For known bounds $\bar Y,\bar W$ and error budget $M\ge0$, the maintained class $\mathcal C$ imposes $|Y|\le\bar Y$, $|U|,|W|\le\bar W$, and the error and positivity restrictions below.
A law $Q$ of $X^\star$ is compatible exactly when
\begin{equation}
 Q_{Y,W}=P_X,\qquad \E_Q\epsilon=0,\qquad
 \E_Q\epsilon^2\le M,\qquad \E_QU^2>0.
 \label{app:eq:regression-model}
\end{equation}
Here $W=U+\epsilon$, and $Q_{Y,W}=P_X$ is the condition $f_\#Q=P_X$.
The error may depend on the outcome and true regressor.
The target and moment are
\[
 \theta_0=\frac{\E_{P_{X^\star}}[UY]}{\E_{P_{X^\star}}U^2},\qquad
 g_\theta(X^\star)=U(Y-\theta U),\qquad \Theta=\mathbb R.
\]
The positive second moment gives a unique slope for each compatible law.

Keep the full state $X^\star$ in the optimization.
For $x^\star=(y,u,\epsilon)$ and $x^{\star\prime}=(y',u',\epsilon')$, write $w=u+\epsilon$ and $w'=u'+\epsilon'$, and use
\begin{equation}
 \begin{split}
 d(x^\star,x^{\star\prime})&=|y-y'|+|u-u'|+|w-w'|,\\
 \mathcal H&=\{Q:Q_{Y,W}=P_X,\ \E_Q\epsilon=0,
                  \ \E_Q\epsilon^2\le M\},
 \qquad \mu=\E_{P_X}W.
 \end{split}
 \label{ex:regression-feasibility}
\end{equation}
Laws in $\mathcal H$ have the stated bounded support and $\E_QU=\mu$.
Assume $\E_{P_X}|W|>\sqrt M$.
Every $Q\in\mathcal H$ satisfies
\begin{equation}
 \E_QU^2\ge(\E_Q|U|)^2
 \ge(\E_{P_X}|W|-\sqrt M)^2=:\kappa>0.
 \label{app:eq:regression-denominator}
\end{equation}
The positivity restriction is therefore automatic, and $\mathcal H=\mathcal F$.
This retains existing identifying information without imposing additional assumptions.

Use the shrinkage imputations $c_\alpha(w)=(1-\alpha)w+\alpha\mu$, with $0\le\alpha\le1$ and $\alpha^2\E_{P_X}(W-\mu)^2\le M$.
Then
\begin{equation}
 Q_\alpha=\mathcal L_{P_X}
       (Y,c_\alpha(W),W-c_\alpha(W))\in\mathcal H,
 \qquad Q_0=\mathcal L_{P_X}(Y,W,0).
 \label{app:eq:regression-full-benchmarks}
\end{equation}
For every $Q\in\mathcal H$,
\begin{equation}
 W_1(Q,Q_\alpha)=\E_Q|U-c_\alpha(W)|.
 \label{app:eq:regression-transport-identity}
\end{equation}
Indeed, coupling each $(Y,U,\epsilon)$ with $(Y,c_\alpha(W),W-c_\alpha(W))$ gives the upper bound.
The function $|u-c_\alpha(w)|$ is $1$-Lipschitz for $d$ and vanishes on the benchmark, giving the reverse bound.
The transport constraint is therefore a single expected-cost budget.
In particular, $W_1(Q,Q_0)=\E_Q|\epsilon|\le\sqrt M$ for every compatible $Q$, so $\delta_0=\sqrt M$ ensures coverage.

More generally, every calibrated pair satisfies
\begin{equation}
 \delta\ge r^\star(Q_\alpha)
 \quad\Longrightarrow\quad
 \mathcal A_\delta(Q_\alpha)=\mathcal H=\mathcal F.
 \label{app:eq:regression-envelope-exact}
\end{equation}
Hence the transport constraint can be dropped from Step~4, and the DRO set equals the sharp identified set.
Benchmark selection can still reduce the radius and improve the Lipschitz tube.
In particular,
\begin{equation}
 |g_\theta(x^\star)-g_\theta(x^{\star\prime})|
 \le\bar W|y-y'|+(\bar Y+2\bar W|\theta|)|u-u'|
 \le L(\theta)d(x^\star,x^{\star\prime}),
 \label{app:eq:regression-lipschitz}
\end{equation}
where $L(\theta)=\max\{\bar W,\bar Y+2\bar W|\theta|\}$.
The general Step~5 tube is
\begin{equation}
 |\E_{Q_\alpha}[YU]-\theta\E_{Q_\alpha}U^2|
 \le\delta L(\theta).
 \label{ex:regression-tube}
\end{equation}

\subsection{Computing the residual moment in Step 4}
\label{app:regression-dual}
Retaining $P_X$ fixes the observable contribution to the moment.
With $h_\theta(Y,W)=W(Y-\theta W)$ and $e_\theta(X^\star)=\epsilon(2\theta W-Y)-\theta\epsilon^2$, substitution of $U=W-\epsilon$ gives, for every $Q\in\mathcal H$,
\begin{samepage}
\begin{equation}
 \E_Qg_\theta=\E_{P_X}h_\theta+\E_Qe_\theta.
 \label{eq:pop-eiv-residual}
\end{equation}
\previewequationalias{pop:eiv-residual-identity}
\par
\end{samepage}
Hence Step~4 adds the known $\E_{P_X}h_\theta$ to the minimum and maximum residual expectations over the selected envelope, as in Appendix~\ref{app:population-residuals}.
These are exactly the original moment bounds.
The following dual computes them for continuous as well as discrete observed laws.

\Needspace{6\baselineskip}
\begin{proposition}[Regression moment computation]
\label{app:prop:regression-dual}
For any calibrated pair in \eqref{app:eq:regression-envelope-exact}, the upper moment bound is
\begin{equation}
 \overline\psi_\theta
 =\inf_{\substack{\eta\ge0\\\nu\in\mathbb R}}
 \left\{\eta M+
 \E_{P_X}\max_{|u|\le\bar W}
 [u(Y-\theta u)-\nu(W-u)-\eta(W-u)^2]\right\}.
 \label{app:eq:regression-constrained-dual}
\end{equation}
Applying the same formula to $-g_\theta$ and negating gives the lower bound $\underline\psi_\theta$.
Both are convex minimizations with analytic inner maxima.
\end{proposition}

\begin{proof}
By \eqref{app:eq:regression-envelope-exact}, optimize over laws of $(Y,W,U)$ with $(Y,W)$ marginal $P_X$, $\E_Q(W-U)=0$, and $\E_Q(W-U)^2\le M$.
Multipliers $\nu$ and $\eta\ge0$ enforce the last two restrictions.
Compact minimax, as in Theorem~\ref{main:tube}, gives the dual; maximization conditional on $(Y,W)$ gives its displayed form.

For fixed $(y,w)$, the inner expression is quadratic in $u$.
Its maximum occurs among
\begin{equation}
 -\bar W,\quad\bar W,\qquad
 \frac{y+\nu+2\eta w}{2(\theta+\eta)}.
 \label{app:eq:regression-dual-candidates}
\end{equation}
The last candidate applies only when $\theta+\eta>0$ and it lies in $[-\bar W,\bar W]$.
For $-g_\theta$, replace $y$ and $\theta$ by $-y$ and $-\theta$.
Thus both bounds require only a convex optimization over two multipliers, observed-law integration, and analytic inner maximization.
No grid for the true regressor is needed.
\end{proof}

Under \eqref{app:eq:regression-denominator}, the affine endpoint result in Appendix~\ref{app:affine-computation} obtains parameter bounds by searching for the zeros of these moment bounds.
Thus the same program supplies both residual optimization and Algorithm~\ref{alg:main-dro-construction}'s parameter search.

\subsection{Applying Algorithm~\ref{alg:main-dro-construction}}
\label{app:recentering-eiv}
We apply the five steps to a continuous observed law.

\paragraph{Numerical data-generating process.}
Let
\begin{equation}
 \begin{aligned}
 U&\sim\tfrac12\operatorname{Unif}[-0.9,-0.882]
       +\tfrac12\operatorname{Unif}[0.882,0.9],\\
 \epsilon&=U/9,\qquad W=U+\epsilon,\\
 Y\mid U&\sim\operatorname{Unif}
       [\tfrac23U-0.4,\;\tfrac23U+0.4].
 \end{aligned}
 \label{app:eq:regression-continuous-dgp}
\end{equation}
Then $|Y|,|U|,|W|\le1$, $\theta_0=2/3$, and the error magnitude $|\epsilon|=|U|/9$ is about $11\%$ of $|U|$.
Set $\bar Y=\bar W=1$ and $M=1/100$; $\E\epsilon^2=7351/750000\approx0.009801<M$.
The generating relations specify one compatible law; identification imposes only \eqref{app:eq:regression-model}.
The observed law has zero means and
\begin{equation}
 \E_{P_X}|W|=\frac{99}{100},\qquad
 \E_{P_X}W^2=\frac{7351}{7500},\qquad
 \E_{P_X}[WY]=\frac{7351}{12500}.
 \label{app:eq:regression-continuous-moments}
\end{equation}
The recorded regressor lies near its support boundaries, permitting radius reductions through inward benchmark movement.

\paragraph{Steps 1--2: initialize and propose benchmarks.}
Use $Q_0$ with $\delta_0=\sqrt M=1/10$ and initial tube \eqref{ex:regression-tube}.
Here $c_\alpha(w)=(1-\alpha)w$; admissibility requires $\alpha^2(7351/7500)\le1/100$.
By \eqref{app:eq:regression-transport-identity}, $\overline s(\alpha)=99\alpha/100$ and
\begin{equation}
 \begin{split}
 \overline r(\alpha)=r^\star(Q_\alpha)
 &=\max_{Q\in\mathcal F}\E_Q|U-c_\alpha(W)|\\
 &=\inf_{\substack{\eta\ge0\\\nu\in\mathbb R}}
 \left\{\eta M+\E_{P_X}\max_{|u|\le1}
 [|u-c_\alpha(W)|-\nu(W-u)-\eta(W-u)^2]\right\}.
 \end{split}
 \label{app:eq:regression-coverage-dual}
\end{equation}
Compact minimax gives this dual.
Its inner maximum uses two quadratic pieces, and any dual choice certifies an upper radius.
The constrained and unrestricted programs \eqref{rec:unrestricted-proposal} yield
\begin{equation}
 \begin{array}{lll}
 \alpha_\triangle=4/245,
 & B_\triangle=Q_{\alpha_\triangle},
 & \delta_\triangle=1027/12250\approx0.083837,\\[2pt]
 \alpha_{\rm free}\approx0.030744,
 & B_{\rm free}=Q_{\alpha_{\rm free}},
 & \delta_{\rm free}\approx0.077644.
 \end{array}
 \label{app:eq:regression-continuous-proposals}
\end{equation}
Both are admissible.
Their optimality with $\delta_B=\overline r(\alpha)$ is verified below.

\paragraph{Step 3: nesting and selection.}
The triangle proposal satisfies $\delta_\triangle+99\alpha_\triangle/100=1/10$.
The unrestricted proposal has a smaller radius, but its triangle bound is approximately $0.108080>\delta_0$.
Exact acceptance follows from \eqref{app:eq:regression-envelope-exact}:
\begin{equation}
 \Phi(B_{\rm free},\delta_{\rm free};Q_0)
 =\max_{Q\in\mathcal F}\E_Q|\epsilon|
 =\sqrt M=\frac1{10}.
 \label{app:eq:regression-continuous-nesting-test}
\end{equation}
Cauchy--Schwarz gives the upper bound; the compatible law $\epsilon=\operatorname{sgn}(W)/10$ attains it.
Algorithm~\ref{alg:main-dro-construction} therefore selects $(Q_{\mathrm c},\delta_{\mathrm c}) =(Q_{\alpha_{\rm free}},\delta_{\rm free})$, with
\begin{equation}
 \mathcal A_{\delta_{\mathrm c}}(Q_{\mathrm c})
 =\mathcal A_{1/10}(Q_0)=\mathcal F.
 \label{rc:eiv-nesting}
\end{equation}
The radius falls from $0.1$ to approximately $0.077644$, while \eqref{app:eq:regression-denominator} gives $\kappa=7921/10000$ throughout.

\paragraph{Step 4: compute parameter bounds.}
Use the residual identity \eqref{eq:pop-eiv-residual} and dual \eqref{app:eq:regression-constrained-dual} for both moment signs.
Symmetrizing candidate laws under $(Y,W,U)\mapsto(-Y,-W,-U)$ preserves their moment and squared error and enforces $\E_Q\epsilon=0$; hence $\nu=0$ suffices.
Exact integration of the piecewise-polynomial conditional maxima and searches for the moment-bound zeros give
\begin{equation}
 \begin{aligned}
 \Theta_{\mathrm{DRO}}(1/10;Q_0)
 &=\Theta_{\mathrm{DRO}}(\delta_{\mathrm c};Q_{\mathrm c})
   =\Theta_{\mathrm{ID}}=[\theta_-,\theta_+],\\
 [\theta_-,\theta_+]&\approx[0.5890455,\;0.6718791].
 \end{aligned}
 \label{app:eq:regression-sharp-set}
\end{equation}
\previewequationalias{app:eq:regression-continuous-outer-bounds}
The witnesses below establish sharpness.
Recentering leaves this set unchanged because both envelopes equal $\mathcal F$.

\Needspace{8\baselineskip}
\paragraph{Step 5: refine and intersect the tubes.}
Since $|Y-3W/5|\le2/5$, writing $c=(1-\alpha)W$ gives
\[
 U(Y-\theta U)-c(Y-\theta c)
 =(U-c)\{Y-\theta(U+c)\},
\]
whose second factor is bounded in absolute value by $2/5+|3/5-(1-\alpha)\theta|+|\theta|$.
Thus \eqref{app:eq:regression-transport-identity} implies
\begin{equation}
 \begin{aligned}
 C_\alpha(\theta)&=\frac25+
       \left|\frac35-(1-\alpha)\theta\right|+|\theta|,\\
 |\E_Qg_\theta-\E_{Q_\alpha}g_\theta|
     &\le C_\alpha(\theta)W_1(Q,Q_\alpha),\qquad Q\in\mathcal H.
 \end{aligned}
 \label{app:eq:regression-benchmark-sensitivity}
\end{equation}
This benchmark-specific bound satisfies $C_\alpha(\theta)\le1+(2-\alpha)|\theta|\le L(\theta)$.
The resulting refined tube is
\begin{equation}
 \widetilde{\mathcal T}_\delta(Q_\alpha)
 =\left\{\theta:
 \left|(1-\alpha)\frac{7351}{12500}
       -\theta(1-\alpha)^2\frac{7351}{7500}\right|
       \le\delta C_\alpha(\theta)\right\}.
 \label{rc:eiv-tube}
\end{equation}
For calibrated pairs it contains $\Theta_{\mathrm{ID}}$ and lies inside \eqref{ex:regression-tube}.
Solving the two positive-slope branches, meeting at $\theta=3/[5(1-\alpha)]$, excludes negative slopes in both the initial and selected tubes and shows that the selected tube lies inside the initial one.
Their intersection is therefore the selected tube, shown in Figure~\ref{fig:regression-comparison}.

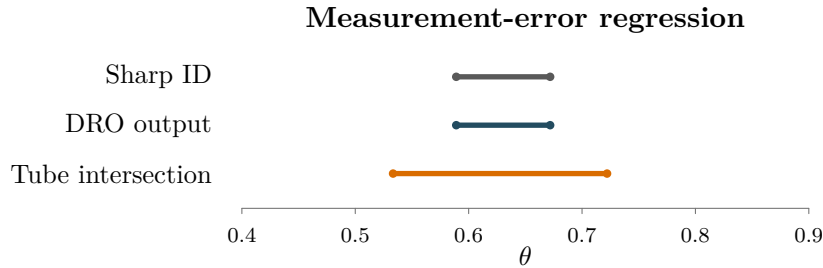
\begin{figure}[htbp]
\centering
\input{figures/measurement_error_regression_bounds}
\caption{Algorithm~\ref{alg:main-dro-construction} outputs for the continuous regression example.
The returned DRO set equals the sharp identified set.
The tube intersection uses the benchmark-specific refinement \eqref{app:eq:regression-benchmark-sensitivity}.}
\label{fig:regression-comparison}
\end{figure}

\Needspace{7\baselineskip}
\paragraph{Verification of the benchmark proposals.}
Write $h=1/50$, $s=1/10$, $v=|W|\sim\operatorname{Unif}[1-h,1]$, and $t=\operatorname{sgn}(W)\epsilon$.
Then $v-1\le t\le v+1$ and $\E_Qt^2\le s^2$; sign symmetry permits $\nu=0$ in \eqref{app:eq:regression-coverage-dual}.
For $0\le\alpha\le(s-h)^2/[4s(1-h)]=4/245$,
\begin{equation}
 |t-\alpha v|\le\frac{t^2}{2s}+\frac s2-\alpha v.
 \label{rc:eiv-radius-certificate}
\end{equation}
For $t\ge\alpha v$, this follows from $(t-s)^2\ge0$; otherwise, from $(s+v-1)^2\ge4s\alpha v$, since $(s+v-1)^2/v$ is minimized at $v=1-h$.
The compatible witness $t=s$ attains the bound, giving $\overline r(\alpha)=r^\star(Q_\alpha)=s-99\alpha/100$.

For $\alpha>4/245$, instead set $t=v-1$ on $v\in[1-h,1-h+hp]$ and choose a positive constant elsewhere to keep $\E_Qt^2=s^2$.
Small $p>0$ gives compatible laws by sign symmetry.
Bound $|t-\alpha v|$ below by $\alpha v-t$ on the first region and $t-\alpha v$ elsewhere.
The resulting lower bound on coverage plus $99\alpha/100$ equals $s$ at $p=0$ and has positive right derivative $2\alpha(1-h)-(s-h)^2/(2s)$.
Thus the triangle-feasible indices are exactly $[0,4/245]$, with minimum radius at the right endpoint.

For the unrestricted optimum, define
\[
 v_*=\sqrt{1-h+h^2/2},\qquad
 t_*=\left\{\frac{hs^2-[h^3-(1-v_*)^3]/3}{1-v_*}\right\}^{1/2}.
\]
Then
\begin{equation}
 \alpha_{\rm free}=\frac{(t_*-1+v_*)^2}{4t_*v_*},\qquad
 \delta_{\rm free}
 =\frac{t_*(1-v_*)+[h^2-(1-v_*)^2]/2}{h}.
 \label{app:eq:regression-free-certificate}
\end{equation}
The compatible law $t=t_*$ for $v>v_*$ and $t=v-1$ otherwise exhausts the budget and satisfies $\E[v\mathbf1\{v>v_*\}]=\E[v\mathbf1\{v<v_*\}]$.
The corresponding signs in $|x|\ge\pm x$ bound coverage below by $\delta_{\rm free}$ for every $\alpha$.
At $\alpha_{\rm free}$, the dual choice $\eta=1/(2t_*)$ gives
\[
 \max\left\{t_*/2-\alpha v,\;
             \alpha v+(1-v)-(1-v)^2/(2t_*)\right\}.
\]
Its branches switch at $v_*$ and are attained by the same law, proving global optimality.

\Needspace{10\baselineskip}
\paragraph{Verification of the sharp endpoints.}
For the upper endpoint, $\E_{P_X}Y^2=76159/187500$ and
\begin{equation}
 \theta_+=\frac{\E_{P_X}[WY]+\sqrt{M\E_{P_X}Y^2}}
                   {\E_{P_X}W^2-M},\qquad
 U_+=\frac{\sqrt{\E_{P_X}Y^2}\,W+\sqrt M\,Y}
                {\sqrt{\E_{P_X}Y^2}+2\sqrt M\,\theta_+}.
 \label{app:eq:regression-sharp-upper}
\end{equation}
This law has mean-zero error, $\E(W-U_+)^2=M$, and $|U_+|<0.956$.
It maximizes $uy-\theta_+u^2-\eta(w-u)^2$ pointwise for $\eta=\sqrt{\E_{P_X}Y^2}/(2\sqrt M)$ and has zero expected moment.
Hence every compatible slope is at most $\theta_+$.

For the lower endpoint, let $\eta_-\approx0.5011841$ be the unique solution of
\[
 \frac1{7500}
 +250\int_{49/50}^1 v\bigl[\tfrac25-(2\eta_--\tfrac35)v\bigr]_+\,dv
 =\frac1{100},
\]
where $[x]_+=\max\{x,0\}$.
Set
\begin{equation}
 U_-=\operatorname{sgn}(W)
       \bigl(1-2\mathbf1\{WY>2\eta_-W^2\}\bigr),
 \qquad \theta_-=\E_{P_X}[U_-Y].
 \label{app:eq:regression-sharp-lower}
\end{equation}
The budget equation and sign symmetry give compatibility.
Since $\eta_-<\theta_-$, convexity makes $U_-$ maximize $\theta_-u^2-uy-\eta_-(w-u)^2$ on $[-1,1]$.
Thus $\E_Qg_{\theta_-}\ge0$ for every compatible $Q$.
Both endpoint laws attain zero; their mixtures attain every intermediate slope.

%% file: figures/measurement_error_regression_bounds.tex
\begin{tikzpicture}[x=1cm,y=1cm,font=\footnotesize,line cap=round]
\node[font=\small\bfseries] at (7.25,.75) {Measurement-error regression};

\node[anchor=east] at (3.25,0) {Sharp ID};
\draw[black!65,line width=2pt] (6.335682,0) -- (7.578187,0);
\fill[black!65] (6.335682,0) circle (1.6pt);
\fill[black!65] (7.578187,0) circle (1.6pt);

\node[anchor=east] at (3.25,-.64) {DRO output};
\draw[inkblue,line width=2pt] (6.335682,-.64) -- (7.578187,-.64);
\fill[inkblue] (6.335682,-.64) circle (1.6pt);
\fill[inkblue] (7.578187,-.64) circle (1.6pt);

\node[anchor=east] at (3.25,-1.28) {Tube intersection};
\draw[orange!85!black,line width=2pt] (5.499890,-1.28) -- (8.331035,-1.28);
\fill[orange!85!black] (5.499890,-1.28) circle (1.6pt);
\fill[orange!85!black] (8.331035,-1.28) circle (1.6pt);

\draw[black!55,line width=.4pt] (3.5,-1.74) -- (11,-1.74);
\foreach \x/\v in {3.5/0.4,5/0.5,6.5/0.6,8/0.7,9.5/0.8,11/0.9}{
  \draw[black!55] (\x,-1.74) -- (\x,-1.81);
  \node[anchor=north,font=\scriptsize] at (\x,-1.86) {\v};
}
\node at (7.25,-2.37) {$\theta$};
\end{tikzpicture}